\documentclass[a4paper,11pt]{article}
\usepackage{graphicx,epsfig, enumerate}
\usepackage{stmaryrd,amsmath, amssymb, amsthm}
\usepackage[all,cmtip]{xy}

\makeatletter
\renewcommand\section{\@startsection {section}{1}{\z@}%
                                   {-3.5ex \@plus -1ex \@minus -.2ex}
                                   {2.3ex \@plus.2ex}%
                                   {\normalfont\large\bfseries}}
\renewcommand\subsection{\@startsection{subsection}{2}{\z@}%
                                     {-3.25ex\@plus -1ex \@minus -.2ex}%
                                     {1.5ex \@plus .2ex}%
                                     {\normalfont\bfseries}}

\makeatother

\newcommand{\bea}{\begin{eqnarray}}
\newcommand{\eea}{\end{eqnarray}}
\newcommand{\be}{\begin{equation}}
\newcommand{\ee}{\end{equation}}

\newcommand{\sgn}{\mathrm{sgn}}

\newcommand{\noi}{\noindent}
\newcommand{\non}{\nonumber}

\newcommand{\Tr}{\mathrm{Tr}}
\newcommand{\im}{\mathrm{Im}}
\newcommand{\re}{\mathrm{Re}}

\newcommand{\half}{\textstyle{\frac{1}{2}}}

\newcommand{\CC}{\mathcal{C}}

\newcommand{\CJ}{\mathcal{J}}
\newcommand{\CM}{\mathcal{M}}
\newcommand{\CN}{\mathcal{N}}

\newcommand{\CP}{\mathcal{P}}

\theoremstyle{definition}

\newtheorem{proposition}{Proposition}

\begin{document}
\begin{titlepage}

\begin{center}

\hfill IPhT-T10/027

\vskip 2 cm {\Large \bf Wall-crossing of D4-branes using flow trees} \vskip 1.25 cm
{Jan Manschot}\\
{\vskip 0.5cm  Institut de Physique Th\'eorique,\\
 CEA Saclay, CNRS-URA 2306, \\
91191 Gif sur Yvette, France}

\pagestyle{plain}
\end{center}

\vskip 2 cm

\begin{abstract}
\baselineskip=18pt
\noi The moduli dependence of D4-branes on a Calabi-Yau manifold is
studied using attractor flow trees, in the large volume limit of the
K\"ahler cone. One of the moduli dependent existence criteria of flow
trees is the positivity of the flow parameters along its edges. It is shown 
that the sign of the flow parameters can be determined iteratively as
function of the initial moduli, without explicit calculation of the
flow of the moduli in the tree. Using this result, an indefinite quadratic form, which appears in
the expression for the D4-D2-D0 BPS mass in the large volume limit,
is proven to be positive definite for flow trees with 3 or less
endpoints. The contribution of these flow trees to the BPS partition
function is therefore convergent.  From non-primitive wall-crossing is deduced that the $S$-duality invariant partition
function must be a generating function of  the rational invariants $\bar
\Omega(\Gamma)=\sum_{m|\Gamma} \frac{\Omega(\Gamma/m)}{m^2}$ instead
of the integer invariants $\Omega(\Gamma)$.
\end{abstract}

\end{titlepage}

\baselineskip=19pt

\tableofcontents

\vspace{0.5cm}

\section{Introduction}
The BPS-spectrum of supersymmetric quantum field theories
\cite{Seiberg:1994rs, Gaiotto:2008cd, Gaiotto:2009hg} and supergravity
\cite{Denef:2000nb, Denef:2007vg, Sen:2007vb} depends in an intriguing
way on the moduli of the theory. If moduli cross walls  of marginal
stability, BPS-states can combine or decay without violating
physical conservation laws. As a consequence, the supersymmetric index
$\Omega(\Gamma;t)$ of BPS-states with charge $\Gamma$, is only locally constant and changes discontinuously as
function of the moduli $t$. This is by no means an arbritrary process but happens according to a
 rather rigorous mechanism, whose implications are however not fully understood.   

The moduli dependence of the supergravity BPS-spectrum
appears as the possible decay or formation of multi-center solutions
if the moduli are varied \cite{Denef:2000nb}. This has led to the conjecture that the moduli
dependence of the supergravity spectrum is captured by ``attractor flow
trees'' \cite{Denef:2000nb, Denef:2007vg}. These trees are schematic (in some sense linearized) representations of supersymmetric
solutions, which are much easier to analyse than the full supergravity
solutions. Various
results have  been derived using the flow trees, such as the (semi-primitive) wall-crossing formula
\cite{Denef:2007vg}, and the derivation of BPS spectra
\cite{Denef:2007vg, Denef:2001xn, deBoer:2008fk, Collinucci:2008ht,
  VanHerck:2009ww}. 

The BPS-states of supergravity are represented in string theory as D-branes wrapped
around cycles of a Calabi-Yau 3-fold $X$. From this point of view, one is
interested in the BPS-spectrum of the D-branes, as function of
the moduli of $X$. A fruitful interplay exists between stability of
D-branes and stability in mathematics \cite{Douglas:2000gi,
  Joyce:2008, Kontsevich:2008}. The BPS indices $\Omega(\Gamma;t)$ are conjecturally equal to
the rigorously defined Donaldson-Thomas invariants.

A central object in the study of BPS-states is the partition function, which is the
generating function for the supersymmetric index $\Omega(\Gamma;t)$ of
BPS-states with charge $\Gamma$. The mixed ensemble is most natural
for $\CN=2$ supergravity \cite{Ooguri:2004zv}, with the electric charges in the
canonical ensemble and the magnetic charges in the microcanonical
ensemble. Besides being the generating function of $\Omega(\Gamma;t)$,
it is a useful object to test the validity on the microscopic level of
duality groups. These are for $\CN=2$ supergravity in 4 dimensions the $S$-duality group
$SL(2,\mathbb{Z})$ \cite{Bohm:1999uk}, and the electric-magnetic duality group
$Sp(2b_2+2,\mathbb{Z})$ (or a subgroup) \cite{de Wit:1996ix}. Most desirable is
a partition function which gives at any given point $t$ in moduli
space the BPS indices $\Omega(\Gamma;t)$, and which
captures correctly the changes of the indices if the moduli are
varied. 

This is a rather difficult problem in general. However, one might
construct the partition function using attractor
flow trees from elementary building blocks, the black hole centers which cannot decay.  
Ref. \cite{Manschot:2009ia} studied in this way the contribution to
the partition function of a flow tree with 2 endpoints with D4-D2-D0 charge. The analysis was simplified by restricting to the
large volume limit of a single complexified K\"ahler cone.  
It shows that a certain indefinite quadratic form, which appears in the
expression for the BPS mass in this limit, is positive definite when
evaluated for stable bound states of two constituents, or equivalently
flow trees with 2 endpoints. This implies the
convergence of the contribution to the partition function of these flow trees, which
enumerates only the stable BPS-states at a point $t$ in the moduli space. The generating function does
not preserve $S$-duality, but can be made so by the addition of a
``modular completion'', which (unexpectedly)  also has the effect
of changing it to a continuous function of the moduli. Continuity appeared in the
  literature before in the context of wall-crossing \cite{Joyce:2006pf,Gaiotto:2008cd}.

The current paper extends the approach of 
Ref. \cite{Manschot:2009ia} to flow trees with 3 endpoints. This
solves various conceptual issues for a generalization to any
number of endpoints. The larger flow trees complicate the analysis
considerably, since the existence (or stability) conditions depend on the flow of the moduli throughout
the tree, and are therefore only indirectly determined by the value $t$
of the moduli at ``infinity''.  The most sensitive condition to 
variations of the moduli is the sign of the flow parameters along the
edges of the tree. The flow parameter is a measure for the length of the edge, and therefore required to be
positive for all edges of an existing flow tree. Fortunately, Subsection
\ref{subsec:sugraflowtrees} derives an iterative expression in terms
of $t$ for this sign, without explicit computation of the flow of the moduli along the
edges. Section \ref{sec:d4d2d0} applies this result to BPS D4-branes,
to proof that also for flow trees with 3 endpoints, an indefinite quadratic form is
positivite definite when restricted to stable flow trees, analogously
to the case of 2 endpoints. This again ensures the convergence of the
partition function. It is expected that this property continues to
hold for flow trees with any number of endpoints.  

To incorporate flow trees with equal charges for 2 of the 3 endpoints, one is required to use the semi-primitive wall-crossing
formula. Section \ref{sec:d4d2d0} argues that partition functions
which capture non-primitive wall-crossing can only be compatible with
$S$-duality, if it is a generating 
function of the rational invariants $\bar \Omega(\Gamma;t)=\sum_{m|\Gamma}
\frac{\Omega(\Gamma/m)}{m^2}$ \cite{Joyce:2008} and {\it not} of the integer invariants
$\Omega(\Gamma;t)$. The jumps of the indices in terms of $\bar
\Omega(\Gamma;t)$ are also more easily identified as contributions
from flow trees than in terms of $\Omega(\Gamma;t)$. The contributions
of the primitive and semi-primitive trees are shown to combine nicely into sums over certain lattices. 

Unfortunately, the form of the stability condition for trees with 3
endpoints prevents an easy construction of the modular completion of
its contribution to the partition function analogous to
Ref. \cite{Manschot:2009ia}. The compatibility of these flow trees with $S$-duality is thus not yet
completely shown, but important prerequisites are satisfied. I hope to
address this issue in future work.  

I conclude the introduction with the outline of the paper. Section
\ref{sec:wall-crossing} reviews wall-crossing of BPS-states to render the paper self-contained. It  
reviews in particular the Kontsevich-Soibelman wall-crossing formula,
wall-crossing in supergravity and the split attractor flow
conjecture. It derives an expression for the sign of the flow parameters, without
explicitly calculating the flow of the moduli throughout the tree. 
Section \ref{sec:d4d2d0} applies the general discussion of Section
\ref{sec:wall-crossing} to D4-D2-D0 BPS-states. The main part of the
section deals with the proof that the indefinite quadratic form is
positive definite on the stable spectrum for $N\leq 3$. Subsection 
\ref{subsec:non-primitive} comments on non-primitive wall-crossing,
and why $S$-duality favours the rational invariants $\bar
\Omega(\Gamma;t)$. Section \ref{sec:conclusion} concludes with a short
summary of the results and  discussion.

\section{Wall-crossing and flow trees}
\label{sec:wall-crossing}
\setcounter{equation}{0}
This section reviews briefly stability and wall-crossing of
BPS-states in string theory compactified on a Calabi-Yau 3-fold $X$ (more
information can be found in the references). This compactification preserves $\CN=2$
supersymmetry, such that the only massive BPS states preserve half of the
supersymmetry. We will work in the Type IIA duality
frame, where the electric-magnetic charges of supergravity
correspond to D-branes wrapping even dimensional cycles of $X$. The charges are
combined into a vector $\Gamma=(P^0,P^a,Q_a,Q_0)^\mathrm{T}$, which is
an element of a  $(2b_2+2)$-dimensional symplectic lattice $L$, with
symplectic inner product: 
\be
\label{eq:sympairing}
\left<\Gamma_1,\Gamma_2\right>=-P_1^0\,Q_{0,2}+P_1\cdot Q_2-P_2\cdot Q_1+P^0_2\,Q_{0,1}.
\ee
$\left<\Gamma_1,\Gamma_2\right>$ is often abbreviated to $I_{12}$ in the following.

The $\CN=2$ superalgebra contains a central element, the central charge
$Z:(L,C_X)\to \mathbb{C}$, which associates to every $\Gamma\in L$ and
point of the moduli space $t=B+iJ\in C_X$ (the complexified K\"ahler cone for Type IIA) a complex number
$Z(\Gamma,t)\in \mathbb{C}=\mathbb{R}^2$.  The mass $M$ of a BPS-state is determined
by the central charge: $M=|Z(\Gamma,t)|$. The (not complexified) K\"ahler cone is a
$b_2$-dimensional cone which parametrizes the volumes of even
dimensional cycles of $X$. The boundary of the cone corresponds to
vanishing of the volume of 2-cycles. From the perspective of mirror 
symmetry, it is natural to consider the ``extended K\"ahler
moduli space'' \cite{Aspinwall:1993nu}, which is the union of all
K\"ahler cones of Calabi-Yaus which are birationally
equivalent. These Calabi-Yaus are however not topologically
equivalent, since continuation of the K\"ahler moduli beyond the
boundary of the K\"ahler cone leads to flops of 2-cycles of
$X$.  Although flops do not 
lead to singular physics, we restrict our attention in this paper to $C_X$, corresponding to
topologically equivalent Calabi-Yaus. 

The index $\Omega(\Gamma;t)$ is a measure for the number of
BPS-states. It is defined by a weighted trace over the Hilbert space
$\mathcal{H}(\Gamma,t)$: 
\be
\Omega(\Gamma;t)=\frac{1}{2} \Tr_{\mathcal{H}(\Gamma,t)}\,(2J_3)^2\,(-1)^{2J_3},
\ee
where $J_3$ is a generator of the rotation group Spin(3). The sum over
the Hilbert space shows that $\Omega(\Gamma;t)$ are integers. An important
property of the index is its independence of the string coupling constant
$g_\mathrm{s}$ and the complex structure moduli of $X$ (in Type IIA). Therefore, the
index can be determined and analyzed at finite $g_\mathrm{s}$ or in
the limit $g_\mathrm{s}\to 0$ depending on which regime is better
suited for the analysis. The first regime corresponds to 
4-dimensional supergravity, where many of the BPS-states appear as
(possibly multi-centered) black holes. The limit $g_\mathrm{s}\to 0$
is the D-brane regime, where the BPS-states can often be 
related to mathematical objects. 

As the notation suggests, the Hilbert space $\mathcal{H}(\Gamma,t)$ depends on
$C_X$. The indices $\Omega(\Gamma;t)$ are only locally
constant and may jump across codimension 1 hypersurfaces in the moduli space. These ``walls of
marginal stability'' are determined by the alignment of central charges
of the constituents $Z(\Gamma_1,t)$ and $Z(\Gamma_2,t)$ with
$\Gamma=\Gamma_1+\Gamma_2$ (assuming that $I_{12}\neq 0$, otherwise
the subspaces of the moduli space where the central charges align are
walls of threshold stability), and divide the
moduli space into chambers. Wall-crossing was 
first observed in 4 dimensions in supersymmetric gauge theory
\cite{Seiberg:1994rs}, and later in supergravity \cite{Denef:2000nb,
  Denef:2007vg}.

\subsection{Kontsevich-Soibelman wall-crossing formula}
\label{subsec:KSwc}
Supersymmetric D-brane configurations lend themselves well to more
abstract descriptions like triangulated categories. Within this 
mathematical setting, Kontsevich and Soibelman \cite{Kontsevich:2008}
have proposed a formula which captures changes of the invariants
$\Delta\Omega(\Gamma_1+\Gamma_2;t)$ at a wall of marginal
stability for generic $\Gamma_1$ and $\Gamma_2$. This was an important open problem in physics, where
the jumps of the indices were only known in restricted situations
like semi-primitive charges \cite{Denef:2007vg}  or Seiberg-Witten
theory \cite{Ferrari:1996sv}. By now a lot of evidence exists for the
validity of the KS-formula in generic BPS contexts
\cite{Gaiotto:2008cd, Gaiotto:2009hg, Dimofte:2009bv,Dimofte:2009tm}. We briefly
review the KS-formula here.  

Ref. \cite{Kontsevich:2008} introduces a Lie algebra with generator
$e_\Gamma$ for every charge $\Gamma \in L$. The commutation relations
are given by
\be
[e_{\Gamma_1},e_{\Gamma_2}]=(-1)^{\left<\Gamma_1,\Gamma_2\right>} \left<\Gamma_1,\Gamma_2\right>\,e_{\Gamma_1+\Gamma_2}.
\ee
For every charge $\Gamma$ an element $T_{\Gamma}$ of the Lie group is
defined by
\be
T_\Gamma=\exp \left(-\sum_{n \geq 1} \frac{e_{n\Gamma}}{n^2} \right).
\ee
A sector in $\mathbb{R}^2$ is defined as a region bounded by two rays whose starting point is at the
origin. A sector is strict if the angle between the rays is less then
$180^\circ$. A product $A_V$
of elements $T_\Gamma$ is associated to a strict sector $V\in \mathbb{R}^2$. The clockwise order
of the central charges $Z(\Gamma,t)\in V$ with $\Gamma\in L$,
determines the order of the product:
\be
\label{eq:KSformula}
A_V=\prod^\curvearrowright_{\Gamma\in L,\,Z(\Gamma,t)\in V} T_\Gamma^{\Omega(\Gamma;t)}.
\ee
If the moduli cross a wall of marginal stability, the order of the
central charges changes and therefore likewise the order of the
product. The claim of \cite{Kontsevich:2008} is that the change of the 
$\Omega(\Gamma;t)$ is precisely such that the product
$A_V$ does not change. The commutation relations of $e_\Gamma$ thus
determine the changes of indices if walls are
crossed. 

Note that the form of the wall-crossing formula also suggests that the
invariants $\bar \Omega(\Gamma;t)$, defined by
\be
\label{eq:baromega}
\bar \Omega(\Gamma;t)=\sum_{m|\Gamma}\frac{\Omega(\Gamma/m;t)}{m^2},
\ee
are convenient. These are valued in $\mathbb{Q}$ and are conjecturally
equal to the invariants which are the central topic in the work of
Joyce \cite{Joyce:2008, Joyce:2006pf}. The product formula (\ref{eq:KSformula}) is
in terms of these invariants more simply expressed using the elements
$R_\Gamma^{\bar \Omega(\Gamma;t)}=\exp \left(\bar
  \Omega(\Gamma;t)\,e_\Gamma\right)$. Eq. (\ref{eq:baromega}) can be inverted with the M\"obius inversion formula
\be
\Omega(\Gamma;t)=\sum_{m|\Gamma}\frac{\bar \Omega(\Gamma/m;t)}{m^2}\, \mu\!\left(m\right),
\ee
with $\Gamma$ primitive. The M\"obius function $\mu(n)$ is defined by:
$\mu(1)=1$; if $n>0$ with prime decomposition $n=p_1^{a_1}\dots
p_k^{a_k}$, then $\mu(n)=(-1)^k$, if $a_i=1$ for $i=1,\dots ,k$; and
$\mu(n)=0$ otherwise.

At a generic point of the walls, only the central charges of two
non-parallel primitive charge vectors
$\Gamma_1$ and $\Gamma_2\in L$ align. We denote the chambers on either
site of the wall by $\mathcal{C}_\mathrm{A}$ and
$\mathcal{C}_\mathrm{B}$. To determine the change of the BPS-indices
between $\mathcal{C}_\mathrm{A}$ and
$\mathcal{C}_\mathrm{B}$, one can truncate the product (\ref{eq:KSformula}) to the lattice
generated by $\Gamma_1$ and $\Gamma_2$. The product then becomes
\be
\prod_{\frac{m}{n}\,\mathrm{decreasing}}
T_{(m,n)}^{\Omega((m,n);t_\mathrm{A})}=\prod_{\frac{m}{n}\,\mathrm{increasing}}
T_{(m,n)}^{\Omega((m,n);t_\mathrm{B})},
\ee
where $(m,n)=m\Gamma_1+n\Gamma_2$. Using the Baker-Campbell-Hausdorff formula
\be
e^{tX}e^{tY}=e^{tY}e^{t^2[X,Y]}e^{\frac{1}{2}t^3(\mathrm{ad}\,
  X)^2Y}e^{\frac{1}{2}t^3(\mathrm{ad}\,
  Y)^2X}e^{-\frac{1}{4}t^4[X,[Y,[X,Y]]]}\dots e^{tX},
\ee
with $(\mathrm{ad}\,X)Y=[X,Y]$ and $t\in \mathbb{R}$,
$\Delta\Omega(m\Gamma_1+n\Gamma_2;t)$ can be determined in principle. For
$(m,n)=(1,1)$ one finds the well-known formula
\be
\label{eq:primwallcross}
\Delta\Omega(\Gamma;t)=(-1)^{\left<\Gamma_1,\Gamma_2\right>-1}\left<\Gamma_1,\Gamma_2\right>
\, \Omega(\Gamma_1;t) \Omega(\Gamma_2;t),
\ee
where we assumed that $\left<\Gamma_1,\Gamma_2\right>>0$ and
$\im(Z(\Gamma_1)\bar Z(\Gamma_2))>0$ in $\mathcal{C}_\mathrm{B}$;
$\mathcal{C}_\mathrm{B}$ is thus the stable chamber. A (product) formula is known for semi-primitive wall-crossing
$(m,n)=(1,n)$ from supergravity \cite{Denef:2007vg}, which is consistent with
Eq. (\ref{eq:KSformula}). Eq. (\ref{eq:d0semip}) of Section
\ref{sec:d4d2d0} gives a similar formula, which is adapted for
wall-crossing of D4-D2-D0 BPS-states in the large volume limit. 

The first example of proper non-primitive wall-crossing is for
$(m,n)=(2,2)$.  The KS-formula is now the only tool to 
compute the change  in the index across a wall. To present the result,
it is useful to use nested lists like
$((\Gamma_1,\Gamma_2),((\Gamma_3,\Gamma_4),\Gamma_5))$, which also
play a large role in the discussion on flow trees in Subsection 
\ref{subsec:sugraflowtrees}. We define the following numbers: 
\be
\label{eq:deltatree}
\bar \Omega(\,(\Gamma_1,\Gamma_2)\,;t)=(-1)^{\left<\Gamma_1,\Gamma_2\right>-1}\left<\Gamma_1,\Gamma_2\right>\, \bar
\Omega(\Gamma_1;t)\, \bar \Omega(\Gamma_2;t),
\ee
which carries on to more complicated lists. For example the nested list $((\Gamma_1,\Gamma_2),\Gamma_3)$
leads to: 
\begin{eqnarray}
&&\bar
\Omega(\,((\Gamma_1,\Gamma_2),\Gamma_3)\,;t)=(-1)^{\left<\Gamma_1+\Gamma_2,\Gamma_3\right>+\left<\Gamma_1,\Gamma_2\right>}\left<\Gamma_1+\Gamma_2,\Gamma_3\right>\,\left<\Gamma_1,\Gamma_2\right>\\
&&\qquad\qquad \qquad \qquad \qquad \times\,\bar \Omega(\Gamma_1;t)\, \bar
\Omega(\Gamma_2;t)\, \bar \Omega(\Gamma_3;t).\non
\end{eqnarray}
The jump of the index $\Delta \Omega(2\Gamma_1+2\Gamma_2;t)$ depends
on the indices $\Omega(a\Gamma_1+b\Gamma_2; t_\mathrm{A})$ in
$\mathcal{C}_\mathrm{A}$ with $a,b\in
\left[0,2\right]$. One finds using the KS-formula:
\begin{eqnarray}
\label{eq:22}
&&\hspace{-2cm}\Delta \bar \Omega(\,2\Gamma_1+2\Gamma_2;t_\mathrm{A})= \non\\
&& \bar
\Omega(\,(\Gamma_1,\Gamma_1+2\Gamma_2)\,;t_\mathrm{A})+\bar
\Omega(\,(2\Gamma_1,2\Gamma_2)\,;t_\mathrm{A})+\bar
\Omega(\,((2\Gamma_1+\Gamma_2),\Gamma_2)\,;t_\mathrm{A})\non \\
&&+\half\bar
\Omega(\,(\Gamma_1,(\Gamma_1,2\Gamma_2))\,;t_\mathrm{A})
+\half \bar
\Omega(\,(\Gamma_2,(\Gamma_2,2\Gamma_1))\,;t_\mathrm{A})\\
&&+ \half \bar
\Omega(\,((\Gamma_1,\Gamma_1+\Gamma_2),\Gamma_2)\,;t_\mathrm{A})+\half
\bar \Omega(\,((\Gamma_2,\Gamma_2+\Gamma_1),\Gamma_1)\,;t_\mathrm{A})\non\\
&&+\textstyle{\frac{1}{4}}\bar
\Omega(\,((\Gamma_2,(\Gamma_1,\Gamma_2)),\Gamma_1)\,;t_\mathrm{A}).\non 
\end{eqnarray}
We observe that the jump $\Delta \bar \Omega(2\Gamma_1+2\Gamma_2)$
is packaged conveniently in terms of $\bar \Omega$'s and nested
lists. Flow trees are also classified by nested lists, the terms in
Eq. (\ref{eq:22}) are thus naturally identified with contributions of
the corresponding flow trees. The KS-formula provides the non-trivial
prefactors. Subsection \ref{subsec:non-primitive} comments more on this.  

\subsection{Supergravity and flow trees}
\label{subsec:sugraflowtrees}
At finite string coupling $g_\mathrm{s}$ (such that the 4-dimensional
Newton constant $G_4$ is finite), BPS-states correspond to solutions of the supergravity equations
of motion which preserve half of the supersymmetry. These solutions
often contain various black holes with macroscopic horizons. 
The (K\"ahler) moduli appear in supergravity as massless
scalars. Their values at infinity are imposed as boundary
conditions. They determine the value of the central charge, and
therefore also the stability of bound states. The values of the moduli
are generically not constant throughout a black hole solution, but ``flow'' to special values
determined by the electric-magnetic charge of the black hole, due to
the attractor mechanism \cite{Ferrara:1995ih}. A point of
concern in the attractor mechanism is the possibility of multiple basins of attraction depending on
the values of the moduli at infinity \cite{Moore:1998pn}. Ref. \cite{Denef:2000nb}
explains how this is related to the points in moduli space where the
volume of a 2-cycle of $X$ vanishes. This paper avoids these
singularities by restricting the moduli to a single K\"ahler cone as
explained in the introduction to this section.  

The $\CN=2$ supergravity Lagrangian admits the action of an
$Sp(2b_2+2,\mathbb{Z})$ duality group \cite{de Wit:1996ix}. The relevant subgroup in the
large volume limit are the translations $\mathbb{Z}^{b_2}$ which act by 
\be
\label{eq:periodicity}
\mathbf{K}(k)=\left(\begin{array}{cccc}1 &  & & \\ k^a & \mathbf{1} & & \\ {1\over 2}
  d_{abc}k^bk^c & d_{abc}k^c & \mathbf{1} & \\
{1\over 6} d_{abc}k^ck^bk^c & {1\over 2} d_{abc}k^bk^c & k^a &
1\end{array}\right), \quad k\in \mathbb{Z}^{b_2},
\ee
simultaneously on the charge $\Gamma$ and the period vector $\Pi=(1,t^a,\half
d_{abc}t^bt^c,\frac{1}{6}d_{abc}t^at^bt^c)^\mathrm{T}$. There is in
addition an $SL(2,\mathbb{Z})$ duality group \cite{Bohm:1999uk} which can be related to the IIB
$S$-duality group by a timelike T-duality or the c-map. $S$-duality acts by fractional linear transformations on
$\tau=C_0+\frac{i\beta}{g_\mathrm{s}}$, and interchanges the $B$- and 
$C$-fields. 

A brief review is now given about multi-center supergravity solutions,
before discussing attractor flow trees. The general form of the metric
of a BPS multi-center solution is \cite{Denef:2000nb} 
\be
ds^2=-e^{2U}(dt+\omega )^2+e^{-2U}d\vec{x}^2.
\ee
Since we consider asymptotically flat space-times $\lim_{r\to\infty}
U,\,\,\omega=1$. The evolution of the Calabi-Yau periods in a single
center solution is such that  
\be
\label{eq:flowimZZ}
2\im\left(e^{-U-i\alpha}Z(\Gamma',t)\right)=\sqrt{G_4}\frac{\left<\Gamma,\Gamma'\right>}{r} +2\im\left(e^{-i\alpha}Z(\Gamma',t)\right)_{r=\infty},
\ee
for every charge $\Gamma'\in L$; $\alpha$ is the phase of
$Z(\Gamma,t)$ \cite{Denef:2000nb}. In principle one can solve for the
evolution of the periods and moduli from this equation. The evolution
is often described in terms of the flow parameter $\rho=\sqrt{G_4}/2r$.

More interesting for discussions about stability are solutions
with more centers. Ref. \cite{Denef:2000nb}  shows that the distance
between two centers in a 2-center solution is given by: 
\be
\label{eq:centerdis}
|x_1-x_2|=\sqrt{G_4} \frac{\left<\Gamma_1,\Gamma_2\right>}{2}\frac{|Z(\Gamma_1+\Gamma_2,t)|}{\im
  (Z(\Gamma_1,t)\bar Z(\Gamma_2,t))},
\ee
where the moduli $t$ are evaluated at $r=\infty$. 
The right hand side can be positive or negative depending on the
values of the moduli at infinity. A negative value indicates 
that the BPS-states do not exist at this point of the moduli
space, or in other words that they are unstable. On the other hand,
positivity does not imply stability, since it is not a
sufficient condition for the existence of a full solution to the
supergravity equations of motion. For example, solutions where the central charge
vanishes at a regular point of the moduli space should be disgarded. 
If we assume that this does not happen, and the
existence of the solution depends only on the sign of the right-hand side
of Eq. (\ref{eq:centerdis}), the contribution to the index of the 2-center
solution as function of the moduli can be written as
\cite{Denef:2007vg, Denef:2002ru, Manschot:2009ia}:
\begin{eqnarray}
\label{eq:cont12}
&&\half \left(\sgn(\im(Z(\Gamma_1,t)\bar Z(\Gamma_2,t)))+\sgn(\left< \Gamma_1,\Gamma_2\right>)\right)\\
&&\times (-1)^{\left< \Gamma_1,\Gamma_2\right>-1}\left< \Gamma_1,\Gamma_2\right>\Omega(\Gamma_1)\Omega(\Gamma_2), \non
\end{eqnarray}
with $\sgn(x)$ defined by
\be
\sgn(x)=\left\{\begin{array}{rr} 1, & x> 0,\\ 0, &  x=0, \\ -1, & x< 0. \end{array} \right. 
\ee
Since close to the wall of marginal stability, the supergravity
solution will always resemble a 2-center solution this is consistent
with  Eq. (\ref{eq:primwallcross}). Note that Eq. (\ref{eq:cont12})
gives a non-zero contribution at the wall. 

Using that $e^{-U}\to \sqrt{G_4}|Z(\Gamma,t)|/r$ for
$r\to 0$, one finds from Eq. (\ref{eq:flowimZZ}) that the attractor equations are equivalent to 
\be
\im (Z(\Gamma,t(\Gamma))\bar Z(\Gamma',t(\Gamma)))=-\left<\Gamma',\Gamma\right>
\ee
for every $\Gamma'\in L$. One observes from this equation that if the
moduli at infinity are fixed at the attractor point $t(\Gamma)$, the right-hand side of
Eq. (\ref{eq:centerdis}) can never be positive, and therefore 2-center solutions can not exist.   

To understand all the implications of the supergravity viewpoint to
BPS-stability, one needs to study solutions with more centers, which
becomes quite complicated. Fortunately, the split attractor flow conjecture
\cite{Denef:2000nb, Denef:2007vg} proposes a rather elegant framework
for analyzing the stability of multi-center solutions as function
of the background moduli. The conjecture has on the other hand not
much bearing on those multi-center solutions, whose stability does
not depend on the moduli. The mysterious scaling solutions lie in this
class \cite{Denef:2007vg}. The conjecture does not distinguish such
solutions from single center solutions. We briefly review the
conjecture at this point, following Refs. \cite{Denef:2007vg, Denef:2000ar}.

The central objects of the conjecture are the so-called ``(attractor)
flow trees'', which are simplified, schematic representations of 
supergravity solutions. An example of a flow tree is presented in
Fig. \ref{fig:flowtree}. Its graph is a rooted tree (meaning a
directed tree with all edges directed away from the root vertex, see
e.g. \cite{Diestel:1997}), and corresponds
to a nested list of the total charge $\Gamma$. The nested
list corresponding to Fig. \ref{fig:flowtree} is
$((\Gamma_1,\Gamma_2),((\Gamma_3,\Gamma_4),\Gamma_5))$.\footnote{For
  notational convenience, the $\Gamma$'s, comma's and outer parentheses
  are in the following omitted from the nested lists, thus
  $((\Gamma_1,\Gamma_2),((\Gamma_3,\Gamma_4),\Gamma_5))\to (12)((34)5)$.} 
The vertices are all connected and have generically either one (the leaves) or three edges connect to it. 
The root vertex $v_0$ (drawn at the top in Fig. \ref{fig:flowtree}) corresponds to the sphere at
infinite radius in the supergravity solution, which surrounds the total charge
$\Gamma$. The $N$ bottom vertices (endpoints)
represent black hole centers with charges $\Gamma_i$, $i=1,\dots ,N$
with $\Gamma=\sum_{i=1}^N \Gamma_i$. A tree with
$N$ bottom vertices has $2N-1$ edges and $N-1$ trivalent vertices. We
denote the set of trivalent vertices by $V$, and the set of edges by
$E$. The vertices, edges and charges can obviously be labeled by binary
words, e.g. $RLL$.

It is useful to introduce some notation associated with a trivalent
vertex $v$, for later recursive applications. A vertex which appears
one vertex before $v$ in the tree is denoted by $vU$. The edge between
$vU$ and $v$ is denoted by $e_v$, and the charge along $e_v$ by
$\Gamma_v$. The charge splits at a trivalent vertex $v$:
$\Gamma_v=\Gamma_{vL}+\Gamma_{vR}$;
$\Gamma_{vL}$ goes off to the left and $\Gamma_{vR}$
to the right. 
\begin{figure}[h!]
\centering
\includegraphics[totalheight=10cm]{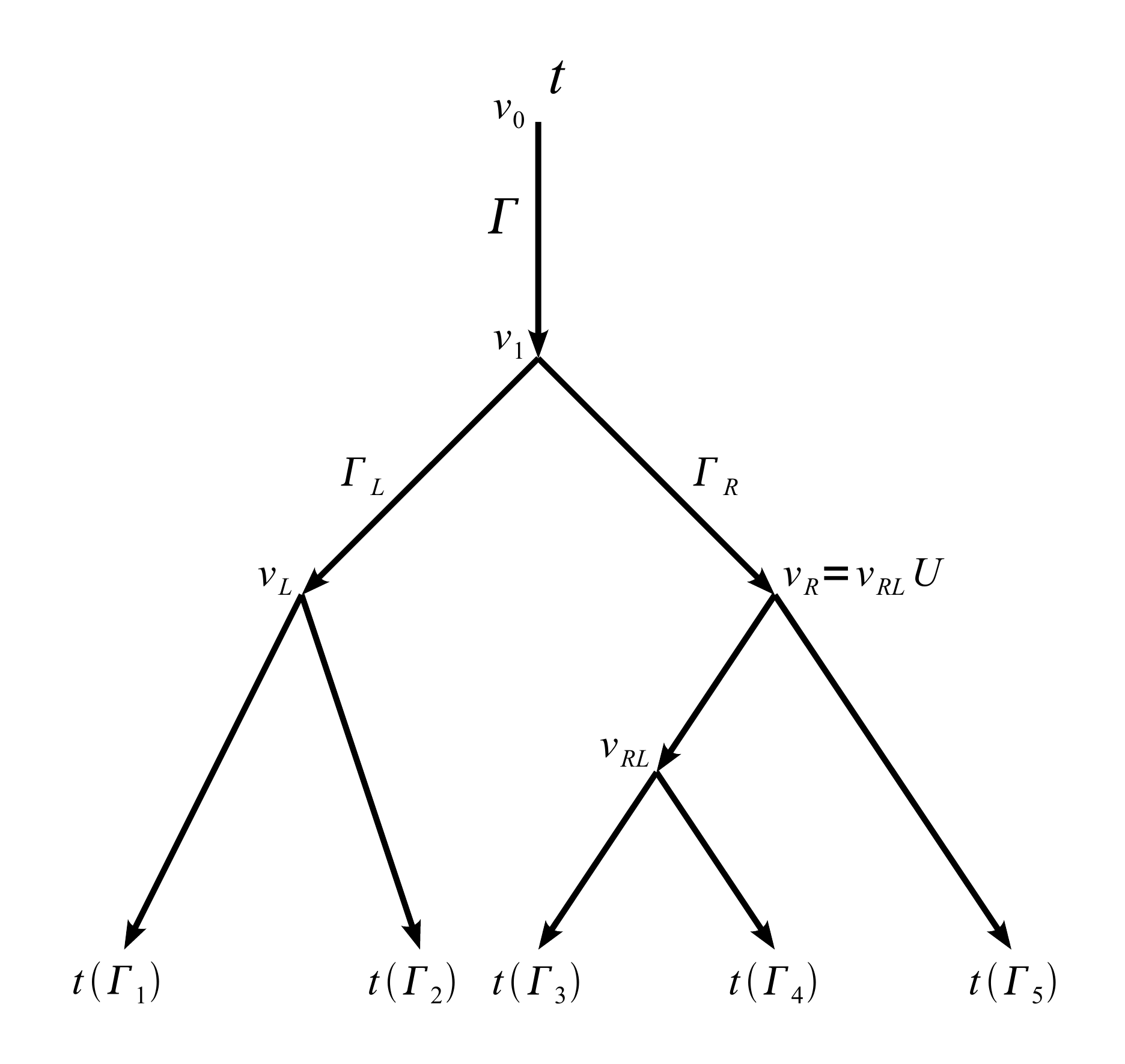}
\caption{The attractor flow tree corresponding to $((\Gamma_1,\Gamma_2),((\Gamma_3,\Gamma_4),\Gamma_5))$.}
\label{fig:flowtree}
\end{figure}

Based on a nested list of charges, one can always construct the rooted
tree. A flow tree is essentially an embedding of the rooted tree $T$ in
moduli space, which might or might not exist depending on the value $t$ of
the moduli at $v_0$. The flow of the moduli along an edge $e\in E$ is given by the evolution of the periods for a single center black
hole (\ref{eq:flowimZZ}) with the corresponding charge  
$\Gamma_{e}$. An edge splits at a trivalent vertex $v$ with modului $t_v$ into edges with charges $\Gamma_{vL}$ and $\Gamma_{vR}$, only if 
$t_v$ is at a wall of marginal stability for $(\Gamma_{vL},\Gamma_{vR})$. If the moduli lie on the intersection of various walls of marginal stability,
the valence of the vertices can increase accordingly. From Eq. (\ref{eq:flowimZZ}), one deduces that the
change of the flow parameter $\Delta\rho_v=\rho_v-\rho_{vU}$ along $e_v$ is:
\be
\label{eq:rhov}
\Delta\rho_v= \frac{\im(Z(\Gamma_{vL},t_{vU})\,\bar Z(\Gamma_{vR},t_{vU}))}{\left<\Gamma_{vL},\Gamma_{vR}
  \right>|Z(\Gamma_{vL}+\Gamma_{vR},t)|}.
\ee
The flows terminate at the bottom vertices, where they are at the
corresponding attractor points $t(\Gamma_i)$.

A flow tree can now be defined more precisely. Given a choice $t$ of
moduli at $v_0$, a flow tree is a rooted tree $T$, which satisfies
the following (stability) conditions \cite{Denef:2000nb}:\\ 
\vspace{-0.5cm}
\label{page.areq}
\begin{enumerate}[A:]
\item 
$\qquad 
\forall\, v\in V:\qquad \qquad \qquad \left<\Gamma_{vL},\Gamma_{vR}
  \right>\,\im(Z(\Gamma_{vL},t_{vU})\,\bar Z(\Gamma_{vR},t_{vU}))> 0. \non\
$
\item 
$\qquad \forall\, v\in V:\qquad \qquad\qquad Z(\Gamma_{vL},t_{v})\,\bar Z(\Gamma_{vR},t_{v})>0. \non$
\item
$\qquad $for $i=1,\dots,N$: $\qquad$  the attractor points
$t(\Gamma_i)$ do exist in the moduli space.
\end{enumerate}
Conditions A and B together imply that $v$ lies at a wall of
marginal stability. Condition A is also equivalent with the positivity
of the flow parameter $\Delta \rho_v$ (\ref{eq:rhov}) along $e_v$. Since it is a
measure for the (inverse) length of the edge, the condition is an
obvious necessary condition for the existence of a supergravity
solution. After all this
introductory material the attractor flow conjecture can be stated:  

\vspace{.3cm}
\noi{\bf Split attractor flow conjecture \cite{Denef:2007vg}:}
\vspace{-.3cm}
\begin{enumerate}
\item components of the
moduli space of (4-dimensional) supergravity solutions with total
charge $\Gamma$ and values of the moduli at infinity $t$, are in 1 to 1
correspondence with flow trees starting with total charge
$\Gamma$ and moduli $t$, 
\item for fixed total charge $\Gamma$ and moduli $t$ only a finite
  number of flow trees exist. By 1. the Hilbert space of
  BPS-states factorizes into a direct sum of the corresponding flow trees. 
\end{enumerate}

This conjecture shows the potential of flow trees
to describe the stability of BPS-states. It suggests an
important role for the endpoints of the flow trees, since these BPS-objects are stable
everywhere in the moduli space. As mentioned before, the endpoints do
not necessarily correspond to a single center, due to the existence of scaling solutions
\cite{Denef:2007vg}. However, the states corresponding to these
endpoints cannot decay at any point in the moduli space. Following 
\cite{Cheng:2007ch}, we will call them ``immortal'' BPS-states. Since
the index of an immortal object with charge $\Gamma$ does not depend on $t$, we simply
denote it by $\Omega(\Gamma)$. The immortal BPS-objects can
thus be found by tuning the moduli to the corresponding attractor
point. In agreement with this, only the  $N=1$ tree exists if
$t=t(\Gamma)$. A convenient aspect of the immortal BPS-objects is that more is
known about their microscopic aspects, their degrees of freedom are
typically those of a conformal field theory, which adds many
symmetries to the problem. 

Whether Condition A is satisfied for $(T,t)$ is conveniently determined by
a product formula:
\begin{eqnarray}
\label{eq:conditionS}
&&\mathrm{\bf Condition\,\, A}:\\
&&\qquad\qquad S(T,t)=\prod_{v\in V} \half
\left(\sgn(\im(Z(\Gamma_{vL},t_{vU})\bar Z(\Gamma_{vR},t_{vU})))+\sgn(\left<\Gamma_{vL},\Gamma_{vR}\right>)\right)\neq
0.\non
\end{eqnarray}
The $\half$ appears in the definition of $S(T,t)$ such that $S(T,t)$ is
 $\pm 1$ instead of $\pm 2^{N-1}$ for flow trees. Similarly, Condition C can be
reformulated as $\prod_{i=1}^N\Omega(\Gamma_i) \neq 0$. Thus, if one
knows that Condition B is satisfied, the contribution of a flow tree
to the index can be found essentially by iteration of
Eq. (\ref{eq:cont12}). The product $S(T,t)$ determines
whether the tree corresponds to (stable) BPS-states, and the
contribution of the flow tree to the index is given by the
KS-formula. Some subtleties arise if multiple endpoints
have equal charges; the next section will comment on this.

Much of the power of the split attractor flow conjecture lies in the
possibility of recursive applications of arguments based on
simple, elementary flow trees. The most elementary rooted tree is
$\Yup$. However, verification of Condition A does not require
determination of the flow of the moduli along its edges.
 This aspect becomes important for the rooted tree corresponding to
$(12)3$, which is displayed in Fig. \ref{fig:flowtree2}. We denote this flow tree by $T_{(12)3}$;
the closely related flow trees with the same total charge are $T_{(23)1}$ and $T_{(31)2}$.
\begin{figure}[h!]
\centering
\includegraphics[totalheight=8cm]{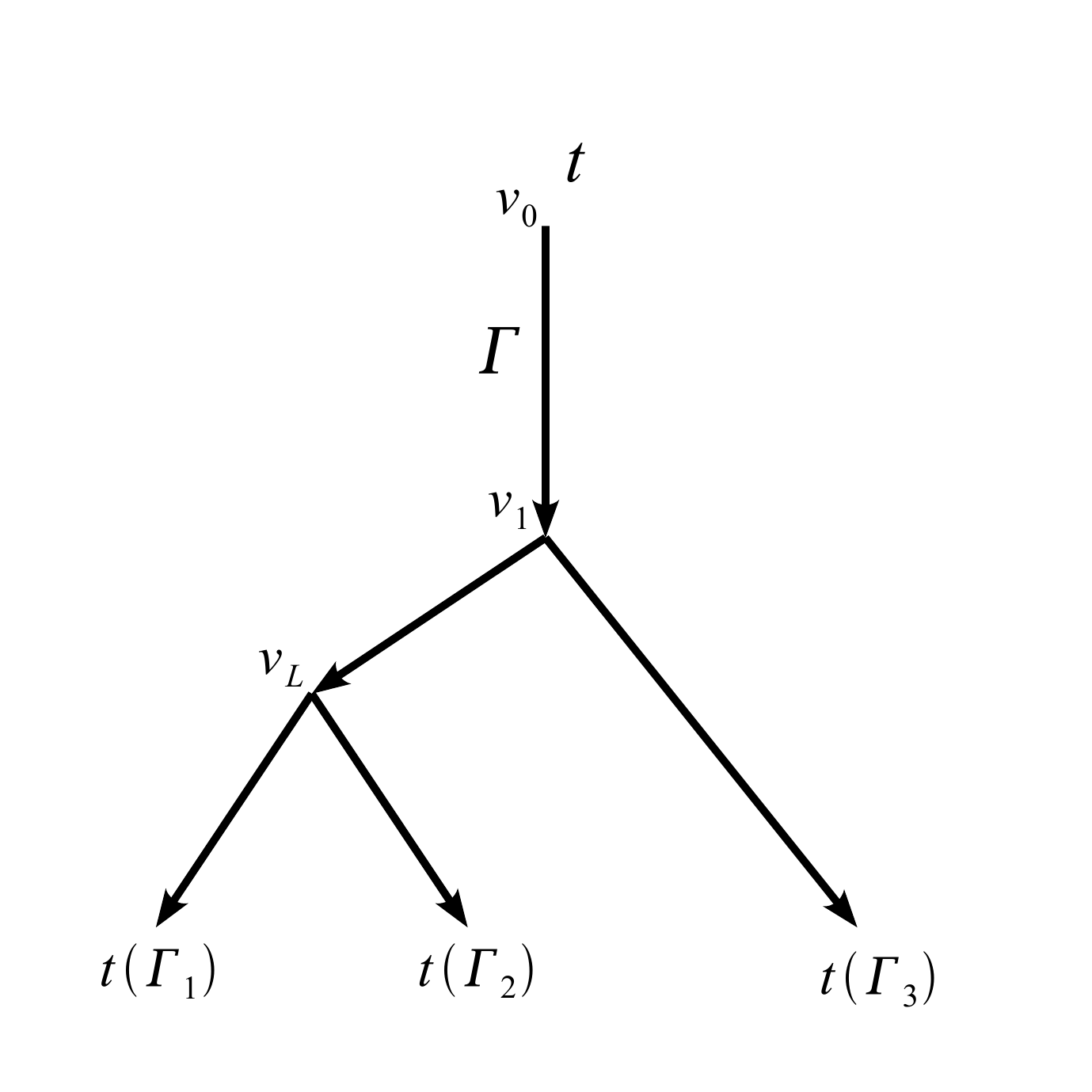}
\caption{Flow tree $T_{(12)3}$ corresponding to $(12)3$.}
\label{fig:flowtree2}
\end{figure}
Assuming that Condition B is satisfied, stability of the split at $v_1$ is determined by
$\sgn\left(\,I_{(1+2)3}\,\im(Z(\Gamma_1+\Gamma_2,t)\bar
Z(\Gamma_3,t))\,\right)$, and similarly the stability of $v_L$ by
$\sgn\left(\,I_{12}\, \im(Z(\Gamma_1,t_1) \bar
  Z(\Gamma_2,t_1))\,\right)$. One might think that the flow of the
periods must be determined explicitly to determine $\sgn\left(\,I_{12}\, \im(Z(\Gamma_1,t_1) \bar
  Z(\Gamma_2,t_1))\,\right)$ in terms of $t$, but this follows
fortunately more directly from Eq. (\ref{eq:flowimZZ}). To see this,  take first $\Gamma'=\Gamma_3$ in
Eq. (\ref{eq:flowimZZ}), which shows that $v_1$ corresponds to the
flow parameter $\rho_1$:
\be
\rho_1=\frac{\im
  (Z(\Gamma_1+\Gamma_2,t)\bar Z(\Gamma_3,t))}{ \left<\Gamma_1+\Gamma_2,\Gamma_3\right>|Z(\Gamma_1+\Gamma_2+\Gamma_3,t)|}.\non
\ee
If one now substitutes $\rho_1$ for $\rho=\sqrt{G_4}/2r$ and $\Gamma'=\Gamma_1$ in
Eq. (\ref{eq:flowimZZ}), and uses that
$Z(\Gamma_1+\Gamma_2,t_1)||Z(\Gamma_3,t_1)$ and $e^U>0$, one finds the desired result
\begin{eqnarray}
\label{eq:imZZ1}
&&\sgn\left( \im\left(Z(\Gamma_1,t_1)\bar Z(\Gamma_2,t_1)
  \right)\right)=\\
&&\sgn \left( \frac{I_{(2+3)1}}{I_{(1+2)3}}\,\,\im\left( 
      Z(\Gamma_1+\Gamma_2,t)\bar Z(\Gamma_3,t)\right)+\,\im\left( Z(\Gamma_1,t)\bar
        Z(\Gamma_2+\Gamma_3,t)\right)\right), \non 
\end{eqnarray}
A more symmetric way of writing
this is
\begin{eqnarray}
\label{eq:stabv1}
&& \sgn\left( \im\left(Z(\Gamma_1,t_1)\bar Z(\Gamma_2,t_1)
  \right)\right)= \\
&&\sgn \left(\sum_{\mathrm{cyclic\,\,permutations\,\, of\,\,} ijk} \frac{I_{(i+j)k}}{I_{(1+2)3}}\,\im \left( Z(\Gamma_i,t)\bar Z(\Gamma_j,t)\right)\right), \non
\end{eqnarray}
which makes more manifest that if $\im\left(Z(\Gamma_1,t_1)\bar Z(\Gamma_2,t_1)
  \right)=0$ all three central charges are aligned. It also shows that
  we have determined the stability at $v_L$ of the two other trees 
  $T_{(23)1}$ and $T_{(31)2}$; the only part which changes is 
$I_{(1+2)3}$. These expressions show that Condition A can be
determined for any flow tree in terms of $t$ in an algorithmic
way. Note that $T_{(12)3}$ can satisfy Condition A, while
$T_{(12)}$ does not if evaluated at $t$. See the discussion on page \pageref{subsubsec:N3}
and further for more details about this for D4-D2-D0 branes. 
If Condition B is satisfied and the splits of the charges are primitive, one can
determine the contribution to the index from this flow tree:
\begin{eqnarray}
\label{eq:cont3endp}
\Omega((12)3;t)=&&\textstyle{\frac{1}{4}}(-1)^{I_{12}+I_{31}+I_{23}}\,I_{(1+2)3}\,I_{12} \,
\Omega(\Gamma_1)\,\Omega(\Gamma_2)\,\Omega(\Gamma_3)\non \\
&& \times\left( \sgn \left(\im(Z(\Gamma_1+\Gamma_2,t)\bar Z(\Gamma_3,t))\right) +
    \sgn(\,I_{(1+2)3})\right) \\
&&\times\left( \sgn\left(\im(Z(\Gamma_1,t_1)\bar
    Z(\Gamma_2,t_1))\right) + \sgn(I_{12})\right). \non
\end{eqnarray}
The contribution of a tree with non-primitive splits has probably a very
similar structure. The analysis of Subsections \ref{subsec:KSwc} 
and \ref{subsec:non-primitive} suggests that the $\Omega$'s should be replaced by $\bar
\Omega$'s and that a non-trivial overall factor might appear.

These generic and exact expressions are useful to make generic statements about
attractor flow trees. A non-trivial question is for example whether the indices based
on attractor flow trees only jump when walls of marginal stability for the
total charge $\Gamma$ are crossed, and not when something non-trivial
happens for the subcharges at the relevant trivalent
vertices. This is of course required by physical arguments, although
not completely obvious for flow trees. Ref. \cite{Denef:2007vg} shows that this is indeed the case
in several concrete examples with $N=3$. Using Eq. (\ref{eq:stabv1})
one can show that for $N=3$, the interplay between the
three trees $T_{(12)3}$, $T_{(23)1}$ and $T_{(31)2}$ is such that
the index does not change when the stability of the splits at
$v_{L,R}$ changes. Eq. (\ref{eq:stabv1}) shows that  $\Omega((12)3;t)$ can jump, if 
\be
\label{eq:stabsubtree}
\sgn\left(\sum_{\mathrm{cyclic\,\,permutations\,\, of\,\,} ijk} I_{(i+j)k}\,\im \left( Z(\Gamma_i,t)\bar Z(\Gamma_j,t)\right)\right)
\ee
goes from $\pm 1$ to $\mp 1$ via 0. This is not necessarily a wall of marginal stability
for $\Gamma=\sum_{i=1}^3\Gamma_i$. However, the contributions to the index
of the trees $T_{(23)1}$ and $T_{(31)2}$, respectively
$\Omega((23)1;t)$ and $\Omega((31)2;t)$, are very similar to
$\Omega((12)3;t)$. In particular, they also contain a factor (\ref{eq:stabsubtree}) and will
thus also jump when $\Omega((12)3;t)$ does. To show that 
$\Omega(\Gamma;t)$ does not jump, we have to show that the
coefficient of  the term (\ref{eq:stabsubtree}) in
$\Omega((12)3;t)+\Omega((23)1;t)+\Omega((31)2;t)$ is zero, if
(\ref{eq:stabsubtree}) is zero. One can show that if
(\ref{eq:stabsubtree}) vanishes,
$I_{(1+2)3}\,\im(Z(\Gamma_1+\Gamma_2,t)\bar Z(\Gamma_3,t))$ and the
cyclic permutations have all the same sign; this is generically
true in a neighborhood of the hypersurface where
(\ref{eq:stabsubtree}) is zero.  Since $\sum_{\mathrm{cyclic\,\,
    permutations\,\, of}\,\,ijk}I_{(i+j)k}I_{ij}=0$, the coefficient
of (\ref{eq:stabsubtree}) thus vanishes. Note that it is very important here that the stability of the subtree is evaluated at $v_1$ and not at
$v_0$. This result for $N=3$ can be applied inductively. Thus the index 
determined by attractor flow trees does only jump when walls for the
total charge are crossed.

This derivation essentially ignored Condition B. More precisely put, it assumes that if
one of the trees, say $T_{(12)3}$, 
exists as flow tree at some point in moduli space,
it cannot be true that Condition B is not satisfied for $T_{(23)1}$, if Conditions A and C are satisfied
(and similarly for $T_{(31)2}$). To argue that this is correct, assume that this could be the case, and that at least one of the splits of $T_{(23)1}$ is a
wall of anti-marginal stability. If the moduli are then moved to the point
where (\ref{eq:stabsubtree}) vanishes, $T_{(12)3}$ implies that the
three central charges align for $t_1$, whereas $T_{(23)1}$ implies that some will
anti-align, which is a contradiction.

Another application of Eq. (\ref{eq:imZZ1}) is the analysis of walls of threshold stability, these are walls in
moduli space where the central charges of say $\Gamma_L$ and
$\Gamma_R$ get aligned, with $\left<\Gamma_L,\Gamma_R\right>=0$.
For $N=3$, this is for example $\left<\Gamma_1+\Gamma_2,\Gamma_3\right>=0$ or a cyclic
permutation. Specific examples of such cases are discussed in 
Ref. \cite{deBoer:2008fk}. 

\section{D4-D2-D0 BPS-states}
\label{sec:d4d2d0}\setcounter{equation}{0}
This section applies the generic discussion of the previous section to
D4-D2-D0 BPS-states. One of the aims is to construct a BPS partition
function which correctly captures the moduli dependence. The BPS
partition function of $\CN=2$ supergravity in the mixed ensemble
\cite{Ooguri:2004zv} takes the following approximate form:  
\be
\label{eq:expansion}
\mathcal{Z}(\tau,C,t)= \sum_{Q_A}   \Omega(\Gamma;t)\, \exp\left(-2\pi \frac{\beta}{g_\mathrm{s}}|Z(\Gamma,t)| + 2\pi i C^A\, Q_A\right),\nonumber 
\ee
where $A=0,\dots ,b_2$. We will use $\beta/g_\mathrm{s}=\tau_2$ and $C^0=\tau_1$ in the
following. Part 2 of the split attractor flow conjecture suggests the
decomposition of the partition function by rooted trees $T$: 
\be
\label{eq:partitionsum}
\mathcal{Z}(\tau,C,t)=\sum_{T\in \mathcal{T}_P} \mathcal{Z}_{T}(\tau,C,t).
\ee
In contrast to the previous section, a rooted tree $T$ in this sum
corresponds to a nested list of magnetic charges $P_i^A$ with the electric charge unspecified; $\mathcal{T}_P$
is the total set of trees based on nested
lists of magnetic charge vectors $P^A$ with $\sum_{i=1}^N
P^A_i=P^A$. The partition function enumerates all possible
distributions of electric charge over the endpoints of these rooted
trees, and determines as function of $t$ whether they correspond to actual flow trees
and contribute to the index.  This section will always use trees in
this sense. Thus $T_{(11)}$ is a tree with equal magnetic charge
vectors associated to the endpoints, which can still have a non-zero
contribution to the index depending on the electric charges. 

To proceed, we make two simplifications:
\begin{enumerate}
\item $P^0=0$, such that there is no netto D6-brane
  charge. The reason for this simplification is that the microscopic
  description is much better understood for immortal BPS-objects with $P^0=0$
  than for $P^0\neq 0$ by a lift to M-theory
  \cite{Maldacena:1997de}. The near-horizon geometry of the resulting black
  string is AdS$_3\times S^2$ and the degrees of freedom combine to a
  2-dimensional $\mathcal{N}=(4,0)$ conformal field theory \cite{Minasian:1999qn}.
\item $J\to \infty$, which is the large volume limit of the K\"ahler
  moduli space. In this limit, quantum effects to the geometry do not
  play a role such that (relatively) basic geometric arguments
  generally suffice. The D-branes are well described in this limit as coherent sheaves on subspaces of
  $X$.  
\end{enumerate} 

In the large volume limit the magnetic charge $P$ (or equivalently the
divisor wrapped by the D4-branes) must be positive, since it
represents the support of a coherent sheaf. The BPS-states with $P^0=0$, which correspond to a single AdS$_3$ throat in 5
dimensions (or equivalently M5-brane), appear in 4 dimensions as single centered or as
multi-centered supergravity solutions. In particular, BPS-states
corresponding to the principal or polar terms in the partition function appear
as bound states of D6 and anti-D6 branes \cite{Denef:2007vg}. When the moduli are varied
such bound states might in principal decay. However this cannot
happen in the large volume limit $J\to \infty$. Ref. \cite{deBoer:2008fk} shows that for $t^a=\lim_{\lambda\to \infty}
D^{ab}Q_b+i\lambda P^a$, with $D_{ab}=d_{abc}P^c$, an
uplift to 5 dimensions leads to only a single AdS$_3$ throat. Since in the limit $\lambda\to \infty$ the dependence on
$\lambda$ disappears, this limit is closely related to the attractor point for D4-D2-D0 black
holes, which is: $t(\Gamma)=D^{ab}Q_b+i\sqrt{\hat Q_{\bar 0}/P^3}
P^a$ ($\hat Q_{\bar 0}$ is defined in the next subsection). These findings are consistent with the results in
\cite{Manschot:2009ia}, where an analysis of the partition function 
showed that for $t=\lim_{\lambda\to \infty}D^{ab}Q_b+i\lambda P^a$, $\Omega(\Gamma;t)$ equals the CFT index.

Based on these considerations, one could state that the CFT
states are those BPS-states in 4 dimensions, which cannot decay in the
large volume limit. Since we will work exclusively in the large volume
limit, we will use the word ``immortal'' for the objects which cannot decay in this limit and
omit the $t$-dependence of the index: $\Omega(\Gamma)$. These immortal
objects form of course a bigger class than the objects which are
immortal in the whole moduli space. Note that different electric
charges correspond to different attractor points: $\Omega(\Gamma';t(\Gamma))$ does not
correspond to $\Omega(\Gamma')$ generically.

\subsection{BPS mass and stability}
\label{subsec:bpsmass}
The form of the partition function shows that its convergence is
essentially determined by properties of the mass $|Z(\Gamma,t)|$ and of
the indices $\Omega(\Gamma;t)$. The contribution to the partition
function of a flow tree with a single endpoint is known to be
convergent by CFT arguments. However, it is not
evident that the contributions of flow trees with more endpoints always lead to convergent partition functions. This
subsection proofs that this is the case for flow trees with 1, 2 and
3 endpoints with D4-brane charge, which gives strong evidence that this will continue to
hold for $N>3$.

The central charge $Z(\Gamma;t)$ is for $J\to\infty$ given by 
\be
Z(\Gamma,t)=-\int_Xe^{-t}\wedge \Gamma. \non
\ee
The real and imaginary part of $Z(\Gamma,t)$ for D4-D2-D0 BPS-states
are 
\begin{eqnarray}
\label{eq:Zreim}
\re(Z(\Gamma,t))&=&\frac{1}{2}P\cdot(J^2-B^2)+Q\cdot B -Q_0, \\
\im(Z(\Gamma,t))&=&(Q-BP)\cdot J, \non
\end{eqnarray}
where the triple intersection product $d_{abc}$ is used to contract
vectors. For $P\cdot J^2\gg |(Q-\half B)\cdot B-Q_0|,\,|(Q-BP)\cdot J|$, the
mass takes the form:
\be
\label{eq:mass}
|Z(\Gamma,t)|=\frac{1}{2}P\cdot J^2+(Q-\frac{1}{2}BP)\cdot B-Q_0+(Q-B)_+^2,
\ee
where terms of $\mathcal{O}(J^{-2})$ are omitted. Note that at the attractor point
$t(\Gamma)$, $J$ is never sufficiently large such that
Eq. (\ref{eq:mass}) is a valid approximation for
$|Z(\Gamma,t(\Gamma))|$. The charges $Q_a$ naturally take values in
the lattice $\Lambda^*$, dual to $\Lambda$ which has quadratic form
$D_{ab}=d_{abc}P^c$ and signature $(1,b_2-1)$ by the Hodge index
theorem \cite{Griffiths:1978}. $Q_+^2=\frac{(Q\cdot J)^2}{P\cdot J^2}$
is the projection to a positive definite subspace of $\Lambda\otimes \mathbb{R}$
parametrized by $j=J/|J|$. The positive definite combination
$2Q_+^2 -Q^2=Q_+^2-Q_-^2$ is called the majorant
associated to $j$. Two expressions which are invariant under the
action of $\mathbf{K}(k)$ (\ref{eq:periodicity}) are $\hat Q_{\bar
  0}=-Q_0+\half Q^2$ and $Q_a-d_{abc}B^bP^c$.

Expression (\ref{eq:mass}) is potentially problematic, since
$|Z(\Gamma,t)|-\frac{1}{2}P\cdot J^2$ is not obviously bounded below. This would therefore allow the possibility that addition of
electric charge can result in a decrease of the mass, which is clearly
unphysical. This would also have
the direct consequence that if such states are part of the spectrum, the partition function
(\ref{eq:expansion}) with the electric charges in the canonical
ensemble is not convergent, independent of the growth of the index
(except that it is non-zero). 

To explain the problem more concretely, we consider a rooted
tree with $N$ endpoints, with (possibly non-primitive) charges $\Gamma_i$, $i=1,\dots ,N$. To
every endpoint a lattice $\Lambda_i$ with quadratic form
$D_i=d_{abc}P_i^c$ is associated. By a slight abuse of notation, we
use $P=(P_1,P_2,\dots,P_N)\in\Lambda_1\oplus \Lambda_2\oplus \dots
\oplus \Lambda_N$ in addition to $P=\sum_i^N P_i\in \Lambda$; and
similarly for $Q=(Q_1,Q_2,\dots, Q_N)\in \Lambda_1^*\oplus
\Lambda_2^*\oplus \dots\oplus \Lambda_N^*$. Using the duality
invariant expressions one can write the mass as
\be
\label{eq:boundedb}
\frac{1}{2}P\cdot J^2+(Q-B)_+^2 + \sum_{i=1}^N \hat Q_{\bar 0,i}-\frac{1}{2}(Q_i-BP_i)_{i}^2.
\ee
The attractor endpoints only exist for $\hat Q_{\bar 0,i}\geq
-c_{\mathrm{R},i}/24=-(P_i^3+c_2(X)\cdot P_i)/24$, where
$c_{\mathrm{R},i}$ are the CFT central charges of the endpoints \cite{Maldacena:1997de}. The problem is thus reduced to the fact that
the quadratic form $(Q-B)_+^2 - \sum_{i=1}^N
\frac{1}{2}(Q_i-BP_i)_{i}^2$ is indefinite with signature $(Nb_2-N+1,N-1)$. However, this section will show
that it is positive definite if Condition A is satisfied:
\be
\label{eq:claim}
\mathrm{\bf  Condition\,\,A}\qquad \Longrightarrow \qquad
(Q-B)_+^2 - \sum_{i=1}^N \frac{1}{2}(Q_i-BP_i)_{i}^2 \geq 0,
\ee
thus it is in particular always positive definite for flow trees.

To this end, we start by taking a closer look at Condition A for these
BPS-states. From Eq. (\ref{eq:Zreim}) is clear that the central charge
gets aligned along the positive real axis of the 
$\mathbb{C}$-plane for $J\to \infty$, the infinitesimal angle with the real axis can
nevertheless vary, which leads to interesting wall-crossing
phenomena. For a split $(\Gamma_1,\Gamma_2)$, $I_{12}\,\im( 
Z(\Gamma_1,t)\bar Z(\Gamma_2,t))\geq 0$ becomes for $J\to \infty$
and constituent charges $\Gamma_1=(0,P_1,Q_1,Q_{0,1})$ and $\Gamma_2=(0,P_2,Q_2,Q_{0,2})$:
\be
\label{eq:stabcondition}
I_{12}\, \left(P_1\cdot J^2\,(Q_2-BP_2)\cdot
  J-P_2\cdot J^2 (Q_1-BP_1)\cdot J\right)\leq 0,
\ee
where only the leading order in $J$ is kept. Note that for this
approximation no walls of marginal stability exist for Calabi-Yaus
with $b_2=1$. The stability condition is invariant under rescalings of $J$:
$B+iJ\to B+i\lambda J$  with $\lambda>0$. The space
of variations of Eq. (\ref{eq:Istability}) due to $J$ has therefore $b_2-1$
dimensions, and is essentially a real projective space. Similarly,
variations of $B$ which are proportional to $J$ do not change the stability condition. Thus the total space of
stability conditions in the case of interest has real dimension $2(b_2-1)$.
Since Eq. (\ref{eq:stabcondition})
is either $\pm \infty$ or 0 for $J\to \infty$,
we define a homogeneous function of degree 0: 
\be
\label{eq:Istability}
\mathcal{I}(\Gamma_1,\Gamma_2;t)=\frac{P_1\cdot J^2\,(Q_2-BP_2)\cdot
  J-P_2\cdot J^2 (Q_1-BP_1)\cdot J}{\sqrt{P_1\cdot J^2\,P_2\cdot
    J^2\,P\cdot J^2}}.
\ee
This has the special property that
\be
\mathcal{I}(\Gamma_1,\Gamma_2;t)^2=|Z(\Gamma_1,t)|+|Z(\Gamma_2,t)|-|Z(\Gamma,t)|. \non
\ee

Eq. (\ref{eq:stabcondition}) is reminiscent of the stability
condition for sheaves on surfaces, but already when subleading powers in $J$ are
taken into account, the equivalence between D-branes and coherent
sheaves disappears \cite{Diaconescu:2007bf}. 
Note that for $P_2=\vec 0$, the wall of marginal stability is given by
$Q_2\cdot J=0$. In case $P_2=\vec 0$, $Q_2$ must be a positive vector
in the large volume limit, since it represents the support of a coherent sheaf. Therefore, $Q_2\cdot J$ lies at the
boundary of the K\"ahler cone, and such walls are not crossed, since we restrict ourselves to the K\"ahler
cone. The assumption that the $P_i$ are positive for every endpoint,
as was assumed in writing Eq. (\ref{eq:boundedb}), is thus consistent
with the restriction to this regime of the moduli space. 

For a rooted tree, Condition A can be verified by the product
$S(T,t)$, which can be determined iteratively using Eq. (\ref{eq:imZZ1}). 
To determine the contribution to the partition function of a rooted
tree, also Conditions B and C on page \pageref{page.areq} should be verified. 
The existence of the attractor point of all endpoints (Condition C) is determined by the CFT partition functions, the attractor point exists if $\hat
Q_{\bar 0,i}\geq -c_{\mathrm{R},i}/24$ (note again that for $\hat Q_{\bar 0}<0$
multicenter solutions are required, but they cannot decay in the
large volume limit). Finally, Condition B is essentially assumed by neglecting the  lower orders in $J$ to the
stability condition: $\re( Z(\Gamma,t))\approx\half P\cdot J^2\gg
0$. Alternatively, one can estimate the flow of the moduli as in Ref. \cite{Andriyash:2008it},
to see that in the very large volume limit the central charges will
never be anti-parallel at the vertices.

The remaining part of this subsection will proof implication
(\ref{eq:claim}) for trees with $1,2$ and 3 endpoints, and comment briefly on $N>3$. Also the
contributions to the partition functions of these trees are discussed.   

\subsubsection{One endpoint}
\vspace{-.2cm}
This case is trivial, since the potentially harmful term can be rewritten as
\be
\label{eq:mass2}
(Q-B)_+^2-\frac{1}{2}(Q-B)^2=\frac{1}{2}(Q-B)_+^2-\frac{1}{2}(Q-B)_-^2,
\ee
which is positive definite on $\Lambda$. Before moving on to $N=2$, a couple
properties of the partition function for $N=1$ are reviewed. The partition
function $\mathcal{Z}_{T_1}(\tau,C,t)$ can be written in the following form: 
\begin{eqnarray}
\label{eq:expansion}
\mathcal{Z}_{T_1}(\tau,C,t)&=& \sum_{Q_{0}, \, Q}\nonumber
\Omega(P,Q,Q_{0})\, (-1)^{P\cdot Q} \\
&&\times e\left(-\bar \tau (- Q_{0}+Q^2/2) +  \tau (Q-B)_+^2/2+ \bar \tau
(Q-B)_-^2/2 + C\cdot (Q-B/2)\right), \nonumber 
\end{eqnarray}
where the leading term to the mass in (\ref{eq:mass}) is omitted since
it leads to a modular invariant overall factor. The lower bound of the
mass together with the expected growth of the index imply that the series is convergent. 

The CFT, which describes the degrees of freedom of immortal objects in
the large volume limit, contains a spectral flow symmetry, which implies that the indices $\Omega(P,Q,Q_0)$ only
depend on $\hat Q_{\bar 0}=-Q_0+\half Q^2$, and the representative
$\mu$ of $Q-\half P$ \footnote{The shift by $\half P$
  arises since $Q$ is valued in the shifted lattice $\Lambda^*+\half P$ \cite{Freed:1999vc,
    Minasian:1997mm}.} in the coset $\Lambda^*/\Lambda$ \cite{deBoer:2006vg,
  Gaiotto:2006wm}. This symmetry is also a well-known property of the dual 
supergravity in AdS$_3$ \cite{deBoer:2008fk}. Modularity and spectral flow furthermore imply that the CFT
elliptic genus can be decomposed in a theta function and a
vector-valued modular form $h_{P,Q-\frac{1}{2}P}(\tau)$ \cite{deBoer:2006vg,
  Gaiotto:2006wm}:
\be
\label{eq:vectorvalued}
h_{P,Q-\frac{1}{2}P}(\tau)=\sum_{Q_0} \Omega(P,Q,Q_0)\,q^{-Q_{0}+\frac{1}{2} Q^2},
\ee
which satisfy the special property that
$h_{P,Q-\frac{1}{2}P}(\tau)=h_{P,Q-\frac{1}{2}P+k}(\tau)$ with $k\in
\Lambda$. The definition (\ref{eq:vectorvalued}) can be found in the
existing literature, however Subsection
\ref{subsec:non-primitive} gives evidence for replacing the integer
coefficients $\Omega(P,Q,Q_0)$ by the rational coefficients $\bar \Omega(P,Q,Q_0)$ for compatibility with $S$-duality.

%
%
%
%

\subsubsection{Two endpoints}
\label{subsubsec:twoend}
\vspace{-0.2cm}
This case is dealt with by Ref. \cite{Manschot:2009ia}. The
potentially problematic term is in this case
\be
\label{eq:QJ}
(Q-B)_+^2 -\frac{1}{2}(Q_1-B)_{1}^2-\frac{1}{2}(Q_2-B)_{2}^2.
\ee
To proof that this quantity is positive definite if $S(T_{12},t)\neq 0$ is
satisfied, we can replace $Q_i-BP_i$ by $Q_i$ without loss of generality.
We proceed by writing the quantities in Eqs. (\ref{eq:stabcondition}) and (\ref{eq:QJ}) in terms of vectors in
$(\Lambda_1\oplus \Lambda_2)\otimes \mathbb{R}$, such that we can
apply techniques of Refs. \cite{Zwegers:2000, Gottsche:1996}. Define the unit
vectors $\mathcal{J}_2$, $\mathcal{P}_{12}$ and $s_{12}\in
(\Lambda_{1}\oplus \Lambda_2)\otimes \mathbb{R}$ by 
\begin{eqnarray}
\label{eq:vectors}
&&\mathcal{J}_{2}=\frac{(J,J)}{\sqrt{(P_1+P_2)\cdot J^2}} ,\qquad
\mathcal{P}_{12}=\frac{(-P_2,P_1)} {\sqrt{(P_1+P_2)P_1P_2}},  \\
&&s_{12}=\frac{(-P_2\cdot J^2\,J,P_1\cdot
  J^2\,J)}{\sqrt{(P_1+P_2)\cdot J^2\,P_1\cdot J^2\,P_2\cdot J^2}}. \non
\end{eqnarray}
Innerproducts of these vectors with $Q=(Q_1,Q_2)$ give the familiar
quantities in $S(T_{12})$: $\mathcal{P}_{12}\cdot Q =I_{12}/\,\sqrt{PP_1P_2}$ and $s_{12}\cdot
Q=\mathcal{I}(\Gamma_1,\Gamma_2,iJ)$. These vectors satisfy:
\begin{proposition}
\label{prop:1}
\be
s_{12}\cdot \mathcal{J}_2=0,\quad \mathcal{J}_2\cdot \mathcal{P}_{12}=0,\quad 
s_{12}\cdot \mathcal{P}_{12}\geq 1.
\ee
\end{proposition}

\begin{proof}
The first two identies follow trivially. It is straightforward to show that
the third identity is positive. To show that it is $\geq 1$, notice that
the lattice $\Lambda_{1}\oplus \Lambda_2$ has signature $(2,2b_2-2)$.  The
three vectors $\mathcal{J}_2$, $\mathcal{P}_{12}$ and $s_{12}$ are positive
definite and since $\mathcal{J}_2$ is orthogonal with $s_{12}$ and
$\mathcal{P}_{12}$, they span a lattice with signature $(2,1)$ if they
are all linearly independent. Therefore,
\be
\left|\begin{array}{ccc} 1 & 0  & 0  \\
0  & 1 & s_{12} \cdot \CP_{12} \\
0  & s_{12} \cdot \CP_{12} &
1 \end{array}\right|<0, \non
\ee
which is equivalent to $ s_{12} \cdot \CP_{12} \geq 1$, where equality
only holds if $s_{12}=\CP_{12}$.
\end{proof}
\noi In terms of these vectors, the claim becomes:
\begin{proposition}
\label{prop:2}
 For $Q=(Q_1,Q_2)\in \Lambda_{1}^*\oplus \Lambda_2^*$, $\sgn(s_{12}\cdot Q)-\sgn(\mathcal{P}_{12}\cdot Q)\neq 0$ implies
\be
\label{eq:QJ2}
(Q_1)_1^2+(Q_2)_2^2-(Q\cdot \CJ_2)^2<0.
\ee
\end{proposition}

\begin{proof}
We can assume that $\mathcal{P}_{12}$ and $s_{12}$ are linearly
independent, since otherwise $\sgn(s_{12}\cdot
Q)-\sgn(\mathcal{P}_{12}\cdot Q)=0$. Therefore, $Q$,
$\mathcal{J}_{2}$, $\mathcal{P}_{12}$ and $s_{12}$ span generically a
subspace of $\Lambda_1\oplus\Lambda_2$ with signature $(2,2)$, or else $Q$ is a
linear combination of $\mathcal{J}_{2}$, $\mathcal{P}_{12}$ and $s_{12}$. Therefore,
\be
\label{eq:determinant} 
\left|\begin{array}{cccc} Q^2 & Q\cdot \mathcal{J}_2  & Q\cdot
    \mathcal{P}_{12}  & Q\cdot s_{12}  \\
Q\cdot \mathcal{J}_{2}  & 1 & 0 & 0\\
Q\cdot \mathcal{P} _{12} & 0 & 1 & s_{12} \cdot \CP_{12} \\
Q\cdot s_{12}  & 0 & s_{12}  \cdot  \CP_{12}&
1 \end{array}\right|\geq 0, \non
\ee
which is equivalent to
\be
Q^2-(Q\cdot \CJ_2)^2\leq \frac{(Q\cdot
  \CP_{12})^2 + (Q\cdot s_{12})^2 -2\,Q\cdot \CP_{12}\,Q\cdot s_{12}
  \,s_{12\cdot \CP_{12}}}{1-(s_{12}\cdot \CP_{12})^2}.\non
\ee
Since $\sgn(s_{12}\cdot Q)-\sgn(\mathcal{P}_{12}\cdot Q)\neq 0$
implies $Q\cdot \CP_{12}\,Q\cdot s_{12}\leq 0$, the proposition follows. 
\end{proof}
 
Before we continue with $N=3$, we elaborate a bit more on the
contribution of $N=2$ flow trees to the partition function. To
construct the partition function, first the contribution of the flow
tree to the index must be determined. We assume here
that the magnetic vectors are primitive, such that the primitive 
wall-crossing formula can be used. Subsection
\ref{subsec:non-primitive} comments on the implications of
non-primitive wall-crossing for the partition function.  

Since the D0-brane charges
$Q_{0,i}$ do not appear in the stability condition, the derivation of
the jump becomes somewhat more complicated. To determine the change
between two adjacent chambers $\mathcal{C}_\mathrm{A}$ and
$\mathcal{C}_\mathrm{B}$, the spectrum can be truncated to states with charges
$\Gamma_1=(P_1,Q_1,Q_{0,1})$, $\Gamma_2=(P_2,Q_2,Q_{0,2})$ and
$\Gamma=(P,Q,Q_{0})$ with $(P_1,Q_1)+(P_2,Q_2)=(P,Q)$. Here the
$(P_i,Q_i)$ are kept fixed, but the $Q_{0,i}$ are not since the
wall is independent of $Q_{0(,i)}$. Eq. (\ref{eq:KSformula}) can thus
be truncated to 
\be
\label{eq:primKSproduct}
\prod_{Q_{0,1}}
T_{\Gamma_1}^{\Omega(\Gamma_1)}\,\prod_{Q_{0}}\,T_{\Gamma}^{\Omega(\Gamma;
  t_\mathrm{A})}\,\prod_{Q_{0,2}} T_{\Gamma_2}^{\Omega(\Gamma_2)}=\prod_{Q_{0,2}} T_{\Gamma_2}^{\Omega(\Gamma_2)}\,\prod_{Q_{0}}\,T_{\Gamma}^{\Omega(\Gamma;
  t_\mathrm{B})}\,\prod_{Q_{0,1}} T_{\Gamma_1}^{\Omega(\Gamma_1)}.
\ee
The Lie algebra elements $e_\Gamma$ are central. Using the
Baker-Campbell-Hausdorff formula for this algebra
$e^Xe^Y=e^Ye^{[X,Y]}e^X$, one can derive that the change in the index
across the wall is:
\begin{eqnarray}
\label{eq:DeltaprimitiveII}
\Delta\Omega(\Gamma;t_\mathrm{A}\to t_\mathrm{B})&=&(-1)^{P_1\cdot
  Q_2-P_2\cdot Q_1-1}\left(P_1\cdot Q_2-P_2\cdot
  Q_1\right)\\
&&\times \sum_{Q_{0,1}+Q_{0,2}=Q_0}\Omega(\Gamma_1;t_\mathrm{A})\,
\Omega(\Gamma_2;t_\mathrm{B}). \non
\end{eqnarray}
This change of the index was assumed in Ref. \cite{Manschot:2009ia},
but not derived from the KS-formula.

Since Eq. (\ref{eq:DeltaprimitiveII}) gives the jump of the index
towards the stable chamber, the contribution
$\Omega_{T_{12}}(\Gamma;t)$ of $T_{12}$ to the total index, is given by Eq. (\ref{eq:DeltaprimitiveII}) with the moduli at the right hand side
at the corresponding attractor points. One finds for the generating function
\begin{eqnarray}
h_{T_{12},Q-\frac{1}{2}P}(\tau;t)&=&\sum_{Q_0} \Omega_{T_{12}} (\Gamma;t)\, q^{-Q_0+\frac{1}{2} Q^2} \non\\
&=&\sum_{Q_1+Q_2=Q}\textstyle{\frac{1}{2}}(\,\sgn(\mathcal{I}(\Gamma_1,\Gamma_2;t))-\sgn(I_{12}))\,)\,(-1)^{P_1\cdot Q_2-P_2\cdot
  Q_1} \nonumber  \\
&& \times \,(P_1\cdot
Q_2-P_2\cdot Q_1)\, 
q^{\frac{1}{2}Q^2-\frac{1}{2}(Q_1)^2_1-\frac{1}{2}(Q_2)^2_2} \non \\
&&\times \,h_{P_1, \mu_1}(\tau)\,h_{P_2,\mu_2}(\tau),\non
\end{eqnarray}
where $Q^2$ and $(Q_i)_i^2$ are the quadratic forms based on $P$ and
$P_i$ respectively. $h_{T_{12},Q-\frac{1}{2}P}(\tau;t)$ is not a
vector-valued modular form; however Ref. \cite{Manschot:2009ia}
continues by showing that summing over the D2-brane charges, leads to the
partition function 
\be
\label{eq:ZP1P2}
\mathcal{Z}_{T_{12}}(\tau,C,t)=\sum_{(\mu_{1},\mu_{2})\in \Lambda_{1}^*/\Lambda_1
  \oplus \Lambda_{2}^*/\Lambda_2}
\overline{h_{P_1,\mu_1}(\tau)}\,\overline{h_{P_2,\mu_2}(\tau)}\,\Psi_{(\mu_{1},\mu_2)}(\tau,C,B), 
\ee
with
\begin{eqnarray}
\label{eq:Psi2}
&&\Psi_{(\mu_{1},\mu_2)}(\tau,C,B)=\sum_{{Q_1\in \Lambda_1+\mu_1+P_1/2 \atop
      Q_2\in \Lambda_2+\mu_2+P_2/2}} \, S(T_{12},t)\,I_{12}\,
  (-1)^{P_1\cdot Q_1+P_2\cdot Q_2-1}  \\
&& \quad \times e\left(\tau (Q-B)_{+}^2/2+ \bar \tau (\sum_{i=1,2}(Q_i-B)_{i}^2-(Q-B)_{+}^2)/2+C\cdot (Q-B/2)\right). \non 
\end{eqnarray}
$\Psi_{(\mu_{1},\mu_2)}(\tau,C,B)$ determines which charge
combinations are stable and which are not. It does not transform as a
theta function, but using techniques of indefinite theta functions
\cite{Zwegers:2000}, one can complete it to a function $\Psi_{(\mu_{1},\mu_2)}^*(\tau,C,B)$ which does
transform as a theta function with weight $(\half, b_2+\half)$. We
therefore call $\Psi_{(\mu_{1},\mu_2)}(\tau,C,B)$ a mock Siegel theta
function. Using the completed function, $\mathcal{Z}_{T_{12}}(\tau,C,t)$ transforms precisely as
$\mathcal{Z}_{T_{1+2}}(\tau,C,t)$ (with $T_{1+2}$ the $N=1$ flow tree
with magnetic charge $P_1+P_2$). An intriguing phenomenon of the modular
completion is that it replaces the discontinuity of the partition function
across walls by a continuous transition. One could say that the
discontinuous invariants $\Omega(\Gamma;t)$ are replaced by functions
$\Omega(\Gamma;t,\tau_2)$ of $t$ and $\tau_2$, which approach the 
original invariants in the limit $\tau_2\to \infty$. If this 
structure is valid in general, taking the limit and crossing a wall
between $\CC_\mathrm{A}$ and $\CC_\mathrm{B}$, 
leads to the following commutative diagram: 
\be
\non
\xymatrixcolsep{5pc}
\xymatrix{\Omega(\Gamma;t_\mathrm{A},\tau_2) \ar[d]^{\tau_2\to\infty} \ar[r]^{t_\mathrm{A}\to t_\mathrm{B}} &\Omega(\Gamma;t_\mathrm{B},\tau_2) \ar[d]^{\tau_2\to\infty}\\
\Omega(\Gamma;t_\mathrm{A}) \ar[r]^{\mathrm{KS}} & \Omega(\Gamma;t_\mathrm{B}) }\
\ee
For a better understanding of the way
$\Psi_{(\mu_{1},\mu_2)}(\tau,C,B)$ determines which states are stable and which not, we explain
briefly the concept of indefinite theta functions.

\subsubsection*{Indefinite theta function}
\vspace{-.2cm}
An indefinite theta function sums over part of 
an indefinite lattice, which belongs either to the positive or negative definite part of the
lattice. Typically such sums do not transform as modular forms, but can
be made so in special cases by the addition of a non-holomorphic term
\cite{Zwegers:2000}. The idea is most easily explained by considering a
lattice $\Lambda$ with signature $(1,b_2-1)$ \cite{Gottsche:1996,
  Zwegers:2000}. 

Given two positive vectors $J,\mathcal{P}\in \Lambda$ with $J\cdot
\mathcal{P}>0$, one can proof that the condition $\half(\sgn(J \cdot
Q)-\sgn(\mathcal{P}\cdot Q))\neq 0$ implies that $Q^2<0$. This proof
is completely analogous to the proof of Proposition \ref{prop:2}; just
omit the term with $Q\cdot \CJ_2$ and identify $\CP, J$ with $\CP_{12}$
and $s_{12}$. Figure \ref{fig:indeftheta} displays the lattice points for which the
condition is satisfied for a 2-dimensional lattice with quadratic form
$\left( \begin{array}{ll} -1 & 0 \\ 0 & 1 \end{array}\right)$ (which
is incidentally the intersection form of 2-cycles on $\mathbb{CP}^2$ blown up at a point).
\begin{figure}[h!]
\centering
\includegraphics[totalheight=10cm]{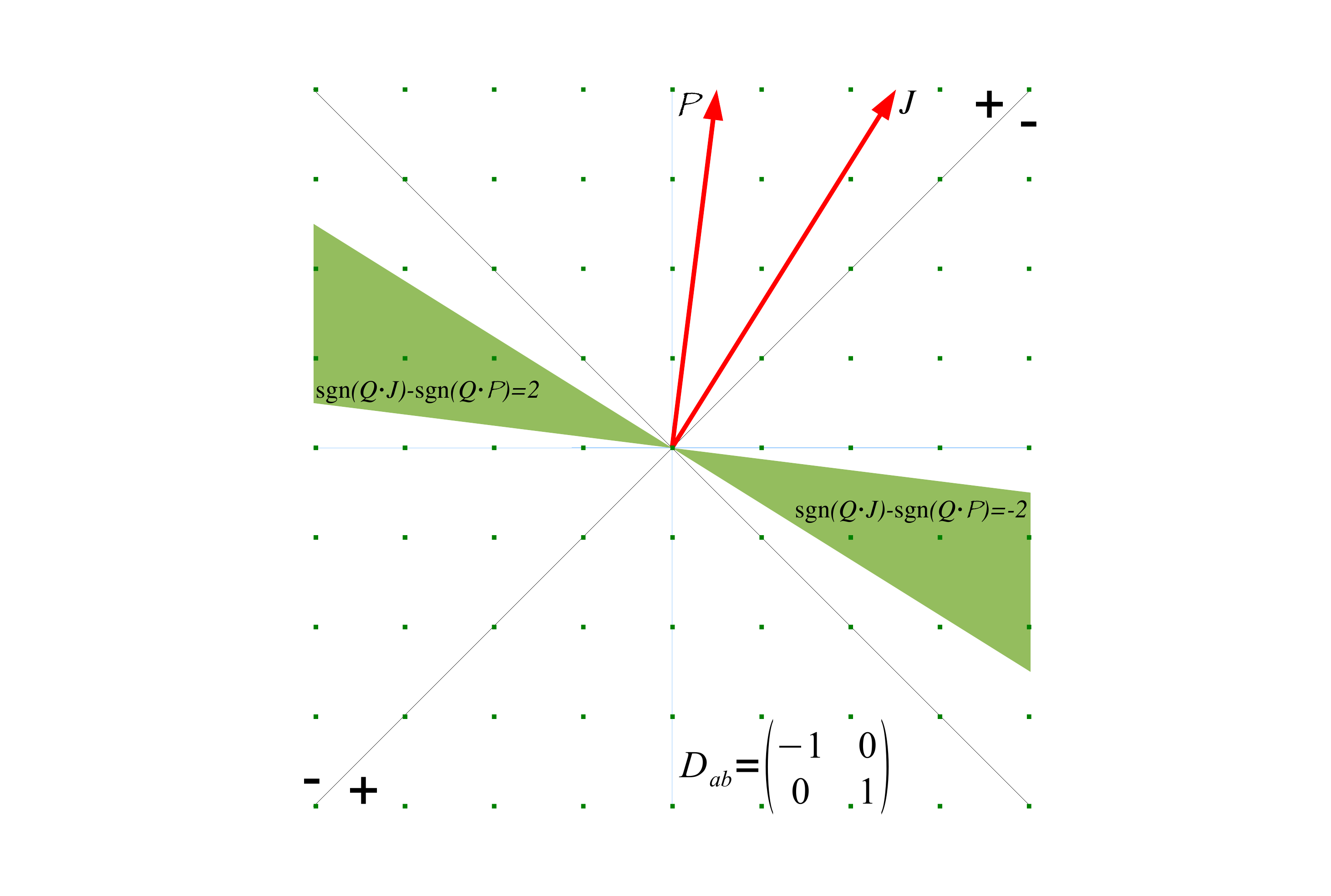}
\caption{An indefinite lattice; the lattice points inside the green
  region contribute to the theta function defined in the text.}
\label{fig:indeftheta}
\end{figure}
The green region in the figure contains the lattice points for which
the condition is satisfied. This region changes when $J$ and/or
$\mathcal{P}$ are varied. (From the point of view of wall-crossing, we
think of $\mathcal{P}$ as fixed and $J$ as variable.)

The indefinite theta function is defined as the sum over all lattice points,
satisfying the condition:
\be
\label{eq:indeftheta}
\theta_\mu(\bar \tau,z)=\sum_{k\in \Lambda} \half(\sgn(J \cdot
Q))-\sgn(\mathcal{P}\cdot Q)) \,\bar q^{k^2/2} y^k,
\ee
which is convergent. Its Fourier coefficients are locally constant as function of $J$, but
can change if the boundary of the green region passes a lattice point. These indefinite theta
functions do not have the nice modular properties which holomorphic theta
functions or Siegel theta functions are known to have. However, the
indefinite theta function can be completed to a function with the
familiar modular properties, by replacing $\sgn(x)$ in (\ref{eq:indeftheta}) by $E(x\sqrt{\tau_2})$
with $E(z)=2\int_0^ze^{-\pi u^2}du$ \cite{Zwegers:2000}. Note that the discontinuous function
$\sgn(Q\cdot J)$ as function of $J$ is replaced now by a continuous function. Moreover,
$E(x\sqrt{\tau_2})$ approaches $\sgn(x)$ for $\tau_2\to \infty$, the
``thickness of the step'' is of order of $\sqrt{2/\tau_2}$. 

The function $\Psi_{(\mu_{1},\mu_2)}(\tau,C,B)$ is very similar to
the function (\ref{eq:indeftheta}). An important difference is that 
the boundary of the positive definite cone depends on the moduli by
$Q\cdot \CJ_2$ in Eq. (\ref{eq:QJ2}). Another
difference is that $\Psi_{(\mu_{1},\mu_2)}(\tau,C,B)$ contains the
factor $P_1\cdot Q_2-P_2\cdot Q_1$ multiplying the exponential,
which leads to a more complicated modular completion.

\subsubsection*{Entropy enigma}
\vspace{-0.2cm}
One can easily compare the relative magnitude of the contribution to the index of flow trees with $N=1$ and
2 using the partition function (\ref{eq:ZP1P2}). A special class is formed by flow trees with $N>1$ whose index exceeds the
index of the flow tree with $N=1$, the so called entropy enigmas. We
consider here entropy enigmas in the Cardy regime of the CFT where
$\hat Q_{\bar 0}\gg P^3$. Ref. \cite{Andriyash:2008it} showed
earlier the existence of entropy enigmas  for D4-D2-D0 branes for weak
topological string coupling $g_{\mathrm{top}}\sim\sqrt{\hat Q_{\bar 0}/P^3}$. The
entropy of the single center is in the Cardy regime:
\be
\pi \sqrt{\frac{2}{3}(P^3+c_2\cdot P)\left(Q_{\bar
      0}+\textstyle{\frac{1}{2}}Q^2\right)}.
\ee
Application of the Cardy formula to Eq. (\ref{eq:ZP1P2}) shows that the
condition for enigmatic $N=2$ flow trees is:
\be
(P^3+c_2\cdot P)\left(Q_{\bar 0}+\textstyle{\frac{1}{2}}Q^2\right)<(P_1^3+P_2^3+c_2\cdot P)\left(Q_{\bar 0}+\textstyle{\frac{1}{2}}(Q_1)_1^2+\textstyle{\frac{1}{2}}(Q_2)_2^2\right).
\ee
Note that the right hand side also captures the entropy due to
distributing the total D0-brane charge in different ways between the two
endpoints, otherwise one should just add up the entropy of both endpoints.

Charges $\Gamma_1$ and $\Gamma_2$, which satisfy this relation, are
not hard to find. To this end, write $Q$ as $\mu -P/2+
k$ with $\mu \in \Lambda^*/\Lambda$ and $k\in \Lambda$. Choose $Q$
such that $k^2=P_1\cdot k^2+P_2\cdot k^2=0$. Therefore, $P_1\cdot k^2=-P_2\cdot
k^2$. Without loss of generality we can assume that $P_1\cdot k^2\geq
0$. Taking $Q_2=0$ leads now to an enigmatic configuration for
sufficiently large $k$. It is not difficult to see that this can very well happen for strong
topological string coupling $g_{\mathrm{top}}\sim\sqrt{\hat Q_{\bar
    0}/P^3}\gg 1$. Substituting this choice of charges into the
stability condition shows that there exist regions in the moduli
space where such bound states are stable. These enigmas show that one has
to be careful by estimating the magnitude of the total index by the
CFT index away from the attractor point.

\subsubsection{Three endpoints}
\label{subsubsec:N3}
This subsection discusses flow trees with three endpoints with
D4-D2-D0 charges. We will proof that also in this case the claim (\ref{eq:claim}) is true, such that the
partition function for flow trees with $N=3$ is convergent. 
The total lattice is now a sum of three lattices:
$\Lambda_1\oplus \Lambda_2\oplus \Lambda_3$. 
The case $N=3$ is qualitatively different from $N=2$,
since the flow of the moduli needs to be taken into account. 
What we want to proof is:
\be
 S(T_{(12)3},t)\neq 0\qquad \Longrightarrow \qquad (Q-B)_+^2 - \sum_{i=1}^3
\frac{1}{2}(Q_i-BP_i)_{i}^2 \geq 0,
\ee
with $S(T_{(12)3},t)$ given by Eqs. (\ref{eq:conditionS}) and (\ref{eq:imZZ1}). 

The requirement that the stability of the subtree $(12)$ is determined
in terms of $t_1$ instead of $t$ has the consequence that the
stability condition is not directly related to a determinant like
Eq. (\ref{eq:determinant}). Therefore, we will reduce $S(T_{(12)3},t)\neq 0$ to special cases where an argument based
on a determinant can be used. To this end, define for generic flow
trees the ``unphysical''  condition:
\begin{eqnarray}
\mathrm{\bf Condition\,\, U}:&& U(T,t)=\prod_{v\in V}\half
\left(\sgn(\left<\Gamma_{vL},\Gamma_{vR}\right>)-\sgn(\mathcal{I}(\Gamma_{vL},\Gamma_{vR},t))\right)\neq
0.\non
\end{eqnarray}
Note that the non-vanishing of $U(T,t)$ is determined here by the
stability of all splits at $v\in V$ in terms of $t$. If stability would
be based on this condition, the jumps of the index might appear at other
points in the moduli space than the walls of marginal stability for
the total charge. It is however a useful condition since:
\begin{proposition}
\label{prop:3}
\be
\label{eq:prop3}
U(T,t)\neq 0\qquad \Longrightarrow \qquad (Q-B)_+^2 - \sum_{i=1}^N
\frac{1}{2}(Q_i-BP_i)_{i}^2 \geq 0,
\ee
\end{proposition}

\begin{proof} It is again sufficient to
  proof the proposition for $B=0$. The vectors defined in
  Eq. (\ref{eq:vectors}), are easily generalized to vectors for vertex
  1 in the tree $T$: $1\to 1L$, and $2\to 1R$. In terms of these vectors, Condition $\bf U$ becomes:
\be
U(T,t)=\prod_{v\in V} \half(\sgn(\mathcal{P}_{vLR}\cdot (Q_{vL},Q_{vR}))-\sgn(s_{vLR}\cdot (Q_{vL},Q_{vR})))\neq 0.
\ee
We will use induction to arrive at the desired result. The proposition is true for $N=2$ by
Proposition \ref{prop:2}. For general $N>2$, the attractor flow tree
can be seen as a combination of two trees $T_{1L}$ and $T_{1R}$
which merge at vertex 1. We index the endpoints of $T_{1L}$ and
$T_{1R}$ respectively by $i=1,2,\dots ,k$ and $i=k+1,\dots ,N$,
such that the left-hand side of the inequality in Eq. (\ref{eq:prop3}) is equal to 
\be
\label{eq:Q2JN}
\frac{(Q_{1L}\cdot
  J)^2}{P_{1L}\cdot J^2}-\sum_{i=1}^k (Q_i)_i^2+\frac{(Q_{1R}\cdot J)^2}{P_{1R}\cdot
  J^2}-\sum_{i=k+1}^N
(Q_i)_i^2-(s_{1LR}\cdot (Q_{1L},Q_{1R}))^2.
\ee
The product $U(T,t)$ factorizes as
\be
U(T,t)=\half (\sgn(\mathcal{P}_{1LR}\cdot
(Q_{1L},Q_{1R}))-\sgn(s_{1LR}\cdot
(Q_{1L},Q_{1R})))\,S(T_L,t)\,S(T_R,t).
\ee
By the induction hypothesis, the sum of the first two terms is positive if 
$S(T_L,t)$ is non-zero, and the similarly the sum of the second two if $S(T_R,t)$ is non-zero. Therefore one
can argue analogously to the proof of Proposition
\ref{prop:2} that $(Q_{1},Q_{2}.\dots,Q_N)$, $\mathcal{J}_{2}$, $\mathcal{P}_{1LR}$ and
$s_{1LR}$ span a space of signature $(2,2)$ in
$\Lambda_1\oplus\Lambda_2\oplus \dots
\oplus \Lambda_N$. Eq. (\ref{eq:Q2JN}) is therefore
negative if $U(T) \neq 0$. 
\end{proof}

For a tree with $N=3$, $S(T_{(12)3},t)\neq 0$  implies in most
cases that $U(T,t)\neq 0$, with $T$ one of the three trees with
$N=3$. Specifically,  $S(T_{(12)3},t)\neq 0$  together with 
\be
\label{eq:s12} I_{12}\left(I_{2(31)}\,P_1\cdot
  J^2+I_{(23)1}\,P_2\cdot J^2\right)\leq 0,
\ee
implies $U(T_{12},t)\neq 0$, and consequently $U(T_{(12)3},t)\neq 0$. To 
analyze the remaining cases, we divide them into three classes:
\begin{eqnarray}
\label{eq:condI}
&{\bf I}&: I_{12}\,I_{31}>0 \quad  \mathrm{and} \quad
I_{12}\,I_{23}<0, \non \\
&{\bf II}&: I_{12}\,I_{31}<0 \quad  \mathrm{and} \quad
I_{12}\,I_{23}>0, \\
&{\bf III}&: I_{12}\,I_{31}>0 \quad  \mathrm{and} \quad
I_{12}\,I_{23}>0. \non
\end{eqnarray}
To proof the positivity for these classes, we only need
to be concerned with those trees for which $S(T_{12},t)=0$ and
$S(T_{(12)3},t)\neq 0$. Then it is possible to show that $\bf I$ implies
$U(T_{2(31)},t)\neq 0$; and similarly that $\bf 
II$ implies $U(T_{1(23)},t)\neq 0$. Class $\bf III$ cannot be
reduced to $U(T,t)\neq 0$ for some $T$, and the proof requires a little more work.  

Let $P=P_1+P_2+P_3$ and define the following unit vectors:
\begin{eqnarray}
\label{eq:vecIII}
&&\mathcal{P}_{12}=\frac{(-P_2,P_1,0)}{\sqrt{(P_1+P_2)P_1P_2}}, \qquad
\mathcal{P}_{23}=\frac{(0,-P_3,P_2)}{\sqrt{(P_2+P_3)P_2P_3}},\non \\
&&\mathcal{P}_{31}=\frac{(P_3,0,-P_1)}{\sqrt{(P_1+P_3)P_1P_3}},\qquad
\mathcal{P}_{(12)3}=\frac{(-P_3,-P_3,P_1+P_2)}{\sqrt{P(P_1+P_2)P_3}},\\
&&s_{(12)3}=\frac{(-P_3\cdot J^2\,J,-P_3\cdot J^2\,J,(P_1+P_2)\cdot
  J^2\,J)}{\sqrt{P\cdot J^2\,(P_1+P_2)\cdot J^2\,P_3\cdot
    J^2}}, \non\\
&&\mathcal{J}_3=\frac{(J,J,J)}{\sqrt{P\cdot J^2}}.\non
\end{eqnarray}
Analogously to Proposition \ref{prop:1}, one
can show various useful relations between these vectors. The innerproduct of
$\CJ_3$ with any other vector in (\ref{eq:vecIII})
vanishes. Furthermore,
\be
\CP_{12}\cdot s_{(12)3}=\CP_{12}\cdot \CP_{(12)3}=0,\quad s_{(12)3}\cdot \mathcal{P}_{(12)3}>1.
\ee

\begin{proposition}
\label{prop:4}
Let $Q=(Q_1,Q_2,Q_3)\in \Lambda_1^*\oplus \Lambda_2^*\oplus
\Lambda_3^*$. If the following conditions are satisfied
\be
\label{eq:condIII}
\begin{array}{l}
a)\quad  (s_{(12)3}\cdot Q)\,(\mathcal{P}_{(12)3}\cdot Q)\geq 0, \\
b)\quad (\mathcal{P}_{12}\cdot Q)\, (\CP_{31}\cdot Q)\geq 0, \\
c) \quad (\mathcal{P}_{12}\cdot Q)\, (\CP_{23}\cdot Q)\geq 0,  \\
\end{array}
\ee
then
\be
\label{eq:QJi3}
\sum_{i=1}^3 (Q_i)_i^2-(Q\cdot \mathcal{J}_3)^2<0.
\ee 
Condition {\it a}) is equivalent to the stability condition for the two
center split $(1+2)3$; Conditions {\it b}) and {\it c}) are equivalent to
Condition $\bf III$ in Eq. (\ref{eq:condI}).
\end{proposition}

\begin{proof}
We start by showing an implication of condition {\it a}) in
(\ref{eq:condIII}). The positive definite subspace of $\Lambda$ is
spanned by the orthonormal basis given by $\CJ$, $\CP_{12}$ and
$\CP_{(12)3}$. Consequently, the vectors $Q$, $s_{(12)3}$, $\CJ$, $\CP_{12}$ and
$\CP_{(12)3}$ span generically a space of signature 
$(3,2)$. Therefore,
\be
\left|\begin{array}{ccccc} Q^2 & Q\cdot \mathcal{J}_3  & Q\cdot
    \mathcal{P}_{12}  & Q\cdot s_{(12)3} & Q\cdot \CP_{(12)3}  \\
Q\cdot \mathcal{J}_3  & 1 & 0 & 0 & 0\\
Q\cdot \mathcal{P}_{12}  & 0 & 1 & 0 & 0\\
Q\cdot s_{(12)3}  & 0 & 0& 1& \CP_{(12)3}\cdot s_{(12)3}
\\
Q\cdot \CP_{(12)3} & 0 & 0& \CP_{(12)3}\cdot s_{(12)3} & 1\end{array}\right|>0. \non
\ee
From this determinant follows that 
\be
Q^2-(Q\cdot \CJ_3)^2-(Q\cdot \CP_{12})^2<0,
\ee
if condition {\it a}) in (\ref{eq:condIII}) is satisfied. Therefore $Q$, $\CJ_3$ and $\CP_{12}$ span in this case a space with
signature $(2,1)$. We want to show that conditions {\it b}) and {\it c}) imply that ``$-(Q\cdot \CP_{12})^2$'' can be
omitted from the inequality. To this end, we choose to complement the
set of three vectors $Q$, $\CJ_3$ and $\CP_{12}$ by 
\be
\CP_{23\bot}=\CP_{23}-(\CP_{23} \cdot \CP_{(12)3})\,\CP_{(12)3},
\ee
which is the component of $\CP_{23}$ orthogonal to $\CP_{(12)3}$. As a
result,  $Q$, $\CJ_3$, $\CP_{12}$ and $\CP_{23\bot}$ span a space of
signature $(2,2)$. Since $\CP_{12}$ and $\CP_{23\bot}$ are both
orthogonal  to $\CJ_3$ and $\CP_{(12)3}$, they span  
a space of signature $(1,1)$. Conditions {\it b}) and {\it c}) imply that $(\CP_{12}\cdot
Q)\,(\CP_{23\bot}\cdot Q)>0$, since 
\be
\CP_{23\bot}=\frac{1}{PP_3(P_1+P_2)}\left( PP_1P_3\,\CP_{23}+PP_2P_3\sqrt{\frac{P_1P_3(P_1+P_3)}{(P_2+P_3)P_2P_3}}\,\CP_{31}\right).
\ee
This also shows that $\CP_{12}\cdot\CP_{23\bot}<0$. Using these
relations together with the argument of the sign of the determinant:
\be
\left|\begin{array}{cccc} Q^2 & Q\cdot \mathcal{J}_3  & Q\cdot
    \mathcal{P}_{12}  & Q\cdot \mathcal{P}_{23\bot}  \\
Q\cdot \mathcal{J}_3  & 1 & 0 & 0\\
Q\cdot \mathcal{P}_{12}  & 0 & 1 & \CP_{12}\cdot \CP_{23\bot} \\
Q\cdot \mathcal{P}_{23\bot}  & 0 & \CP_{12}\cdot \CP_{23\bot} & \CP_{23\bot}^2 
\end{array}\right|>0, \non
\ee
one obtains the desired result
\be
Q^2-(Q\cdot \CJ_3)^2<0.
\ee 
\end{proof}

This proof gives more confidence that positivity can
be proven for any $N$. It is conceivable that for any $N$, $S(T,t)\neq 0$ can be
reduced for most $T$ to $U(T')\neq 0$ for several $T'$, and that in the
remaining cases it can be proved as well. An obstacle for an easy
inductive proof, analogous to the one for $U(T,t)$, is the fact that stability of
subtrees at $v_0$ is not ensured by stability at $v_1$. The quadratic
form for $T_{12}$ is not even positive definite for $S(T_{(12)3})\neq 0$.

Proposition \ref{prop:4} implies that the lattice sum
\begin{eqnarray}
\label{eq:psi3}
&&\Psi_{(\mu_{1},\mu_2,\mu_3)}(\tau,C,B)=\sum_{{Q_1\in \Lambda_1+\mu_1+P_1/2 \atop
      {Q_2\in \Lambda_2+\mu_2+P_2/2 \atop Q_3\in
        \Lambda_3+\mu_3+P_3/2}}} \, S(T_{(12)3},t)\,I_{(12)3}\,I_{12}\,  (-1)^{P_1\cdot Q_1+P_2\cdot Q_2+P_3\cdot Q_3} \non\\
&&\quad \times e\left(\tau (Q-B)_{+}^2/2+ \bar \tau (\sum_{i=1}^3 (Q_i-B)_{i}^2-(Q-B)_{+}^2)/2+C\cdot (Q-B/2)\right),  
\end{eqnarray}
is convergent. Analogously to the discussion in Subsection
\ref{subsubsec:twoend}, this object does not transform as a modular
form. Since it is a lattice sum it is not unlikely that a modular
completion exists for this sum as for $N=2$. This is also expected from
$S$-duality. However, due to the complexity of $S(T_{(12)3},t)$, 
this does not seem as easy as straightforward.  If $S(T_{(12)3},t)$ is
replaced by $U(T_{(12)3},t)$ one can iterate the 
procedure in Ref. \cite{Manschot:2009ia}. We will not attempt to find
the modular completion of Eq. (\ref{eq:psi3}), but leave this for 
future research.  

Nevertheless, we can now write down the
contribution of flow trees with three endpoints to the partition function: 
\be
\label{eq:Z123}
\mathcal{Z}_{T_{(12)3}}(\tau,C,t)=\sum_{(\mu_{1},\mu_{2},\mu_3)\in \Lambda_{1}^*/\Lambda_1
  \oplus \Lambda_{2}^*/\Lambda_2\oplus \Lambda_{3}^*/\Lambda_3}
\overline{h_{P_1,\mu_1}(\tau)}\,\overline{h_{P_2,\mu_2}(\tau)}\,\overline{h_{P_3,\mu_3}(\tau)}\,\Psi_{(\mu_{1},\mu_2,\mu_3)}(\tau,C,B).
\ee
The other topologies of the tree can similarly be taken into
account. If the $P_i$ are primitive and different, the partition
functions for $N=1,2$ and $3$ capture correctly the total jumps of
the indices across walls. We would also like to include the case when
the $P_i$ are possibly equal. In that case one must use the
semi-primitive wall-crossing formula, we will come back to this point
in Subsection \ref{subsec:non-primitive}.

\subsubsection*{Numerical experiments}
\vspace{-.3cm}
Besides the analytical proof of the claim, it is instructive to carry
out numerical experiments to answer questions like: what portion of
the set of rooted trees is a flow tree for given $t$? or what is the overlap between
Conditions A and U. I have done numerical experiments with three Calabi-Yaus, with $b_2=2,3$ and 4. The
Calabi-Yau with $b_2=2$ is discussed in more detail in Ref.
\cite{Candelas:1993dm}, and $b_2=3,4$ in Ref. \cite{Klemm:2004km}.
The only relevant data for our purpose are the triple intersection numbers, 
which are listed in Table \ref{tab:intnumbers}.   

\begin{table}[h]
\caption{Non-zero intersection numbers of  Calabi-Yaus with $b_2=2$ \cite{Candelas:1993dm} and $b_2=3,4$ \cite{Klemm:2004km}. }
\label{tab:intnumbers}
\begin{center}
\begin{tabular}{c|p{2cm}|p{4cm}|p{5cm}}
$b_2$ & 2& 3 & 4 \\
\hline
$d_{abc}$ & $d_{111}=8$, $d_{112}=4$& $d_{111}=8$, $d_{112}=2$, $d_{113}=2$, $d_{123}=1$  &
$d_{112}=4$, $d_{113}=2$, $d_{122}=4$, $d_{123}=2$, $d_{124}=2$, $d_{134}=1$, $d_{224}=2$, $d_{234}=1$
\end{tabular}
\end{center}
\end{table}

Many different tables with combinations of statistical data can be
generated. I suffice here by giving Table \ref{tab:numerics}, which
lists the number configurations with $S(T_{(12)3},t)\neq 0$, the
number for which $U(T_{(12)3},t)\neq 0$, and the number of
configurations which lie in both classes. A C++ 
code has searched $10^9$ configurations per Calabi-Yau, using a
random number generator. The random number generator chose its values 
for the moduli and  the charges in the following domains:
$J^a\in [1,\dots ,12]$, $P^a\in [1,\dots ,10]$, $Q_a\in [-20,\dots
,20]$. The variation of the quantities in the table between different runs of $10^9$ configurations
is $<0.05 \%$. Clearly, the physical condition $S(T_{(12)3},t)\neq
0$ is less often satisfied than the condition $U(T_{(12)3},t)\neq 0$,
although it is not a subset of it. One can also read off from the table, that for all three Calabi-Yaus the ratio of the number of
charge combinations with $T_{(12)3}$ stable, but $T_{12}$ unstable in
terms of $t$ ($S(T_{12},t)=0$), is between 6 and $7\%$. It would be
interesting to better understand the 
dependence on Calabi-Yau, moduli or charges of these and other ratios,
and derive them analytically. 

\begin{table}[h]
\caption{Number of trees in a search of $10^9$ trees $T_{(12)3}$, for which
  $S(T_{(12)3},t)\neq 0$, $U(T_{(12)3},t)\neq 0$ and the number of trees which
  satisfy both conditions.} 
\label{tab:numerics}
\begin{center}
\begin{tabular}{p{1cm}|p{3cm}|p{3cm}|p{3cm}}
$b_2$ & $S(T_{(12)3},t)\neq 0$ & $U(T_{(12)3},t)\neq 0$ &
$S(T_{(12)3},t)\neq 0\quad$  $\bigcap U(T_{(12)3},t)\neq 0$ \\
\hline
$2$ & 18147241 & 29465018 & 17016426\\
$3$ & 22255909 & 35817183 & 20750877\\
$4$ & 23264713 & 37135142 & 21654091\\
\end{tabular}
\end{center}
\end{table}

%

\subsection{Non-primitive wall-crossing}
\label{subsec:non-primitive}

This last subsection discusses some aspects of non-primitive
wall-crossing. Ref. \cite{Denef:2007vg} presents a formula for
the jumps of the index, for semi-primitive wall-crossing
$\Gamma\to N\Gamma_1+\Gamma_2$, which is known to be
compatible with the KS-formula. For the application to D4-D2-D0
BPS-states in the large volume limit, where the walls are independent of $Q_{0(,i)}$, a
wall-crossing formula with an additional parameter for the D0-brane
charge is desired. This formula can be derived from
the KS-formula similar to Ref. \cite{Chuang:2009}. We take the constituent charges to be
$\Gamma_1=(N\gamma_1,Q_{0,1})$ and $\Gamma_2=(\gamma_2,Q_{0,2})$, with $\gamma_1=(P_1,Q_1)$ and
$\gamma_2=(P_2,Q_2)$ respectively. One finds for the
generating series of the indices 
\begin{eqnarray}
\label{eq:d0semip}
&&\sum_{N=0}^\infty\sum_{Q_0}\Delta\Omega((N\gamma_1+\gamma_2,Q_0);t)u^Nv^{Q_0}=\sum_{Q_{0,2}}^\infty\Omega((\gamma_2,Q_{0,2}))v^{Q_{0,2}} \\
&&\qquad
\times\prod_{k=1}^\infty\prod_{Q_{0,1}}\left(1-(-1)^{I_{12}k}u^kv^{Q_{0,1}}\right)^{I_{12}  \,k\,\Omega((k\gamma_1,Q_{0,1}))}.\non
\end{eqnarray}
The $\Delta\Omega(\Gamma;t)$ are the contributions to the index in a
stable chamber for $T_{12}$ with $I_{12}>0$. For $N=1$ one obtains our previous result
(\ref{eq:primwallcross}). One finds for $N=2$: 
\begin{eqnarray}
\label{eq:semiN2}
\Delta\Omega(\Gamma;t)&=&-\sum_{Q_{0,1}+Q_{0,2}=Q_0}2 I_{12}\,\Omega((2\gamma_1,Q_{0,1});t)\,\Omega((\gamma_2,Q_{0,2});t)\\
&&+\sum_{{Q_{0,1}+Q_{0,2}+Q_{0,3}=Q_0 \atop Q_{0,1}\neq
    Q_{0,3}}}I_{12}^2\,\Omega((\gamma_1,Q_{0,1});t)\,\Omega((\gamma_1,Q_{0,3});t)\,\Omega((\gamma_2,Q_{0,2});t)
\non\\
&&+\sum_{2Q_{0,1}+Q_{0,2}=Q_0}\half I_{12}\,\Omega((\gamma_1,Q_{0,1});t)\,\Omega((\gamma_2,Q_{0,2});t)\non
\\
&&\qquad\times\left(I_{12}\,\Omega((\gamma_1,Q_{0,1});t)- 1 \right)\non.
\end{eqnarray}
This expression raises a puzzle. The
discussion of Ref. \cite{Manschot:2009ia} (see the review on page
\pageref{subsubsec:twoend} and further), suggests that a prerequisite
for $S$-duality
invariance of the generating function of $\Delta \Omega(\Gamma;t)$, is
that it can be expressed in terms of products of vector-valued modular forms
of $SL(2,\mathbb{Z})$. However, the $``-1"$ in the last line makes that a factor  $h_{P_1,\mu_1}(2\tau)$
would appear in the current case, which is not a vector-valued modular form of $SL(2,\mathbb{Z})$ but of the congruence
subgroup $\Gamma_0(2)$. The resolution to this puzzle is that the
correct definition of $h_{P,\mu}(\tau)$ is {\it not} as generating function of $\Omega(\Gamma)$ but
instead of  $\bar \Omega(\Gamma)
=\sum_{m|\Gamma}\frac{1}{m^2}\Omega(\Gamma/m;t)$. Requiring that the
newly defined $h_{P,\mu}(\tau)$ transform as an $SL(2,\mathbb{Z})$
vector-valued modular form is compatible with semi-primitive
wall-crossing. To this end, redefine $h_{P,Q-\frac{1}{2}
  P}(\tau)$:   
\be
h_{P,Q-\frac{1}{2} P}(\tau)=\sum_{Q_0} \bar \Omega((P,Q,Q_0))\,q^{Q_{\bar
    0}+\frac{1}{2} Q^2}.
\ee
The generating function of $\Omega(\Gamma)$
transforms only under a congruence subgroup $\Gamma_0(M)$, with $M$ a product of primes $p$: $M=\prod_{p^{\alpha_p}|P}p^{\alpha_p}$, for total magnetic
charge $P$. For $N=2$, it is $h_{2P_1,2\mu_1}(\tau)-\frac{1}{4}h_{P_1,\mu_1}(2\tau)$
which has an expansion with integer coefficients, but does not transform well under $SL(2,\mathbb{Z})$.

Using this new definition, the contribution to the generating function
of \\$\sum_{Q_0}\Delta\Omega(\Gamma;t)\,q^{-Q_0+\frac{1}{2}Q^2}$ in a stable
chamber is:
\be
\sum_{2Q_1+Q_2=Q} \,q^{\frac{1}{2}Q^2-(Q_1)^2_1-\frac{1}{2}(Q_2)^2_2}\times \left(\half I_{12}^2\,h^2_{P_1,\mu_1}(\tau)h_{P_2,\mu_2}(\tau)-\,2I_{12}\,
h_{2P_1,2\mu_1}(\tau)h_{P_2,\mu_2}(\tau) \right). \non
\ee
The two terms can be identified as contributions of the trees $T_{(12)1}$
and $T_{(2\cdot 1\,2)}$.\footnote{The tree $T_{(2\cdot 1\,2)}$ has 
  two endpoints, one with magnetic charge $2P_1$ and one with $P_2$.}
$T_{(12)1}$ should be considered as a special (degenerate) case
of $T_{(12)3}$. We also observe that modularity of the complete partition function,
requires that the $T_{(12)1}$-contribution should combine with a mock Siegel
theta function of the lattice $\Lambda_1\oplus \Lambda_1\oplus
\Lambda_2$, whereas the $T_{(2\cdot 1\,2)}$-contribution should
combine with a mock Siegel theta function of $\Lambda_{2\cdot 1}\oplus
\Lambda_2$ (where $\Lambda_{2\cdot1}$ has quadratic form
$2d_{abc}P_1^c$). Therefore, to show the compatibility of semi-primitive
wall-crossing with modularity, one is forced to understand the extended
flow trees, which we studied before. If we insert the products 
$S(T_{(12)1},t)$ (which is $-\half$ or 0) and $S(T_{(2\cdot1\,2)},t)$, and add the contributions of $T_{(2\cdot1\,2)}$ with primitive charges, we find
\begin{eqnarray}
\label{eq:contftrees}
&&\sum_{Q_1+Q_2+Q_3=Q}
\,q^{\frac{1}{2}Q^2-\frac{1}{2}(Q_1)^2_1-\frac{1}{2}(Q_3)^2_1-\frac{1}{2}(Q_2)^2_2}\, S(T_{(12)1},t)
I_{(12)1}\,I_{12}\,h_{P_1,\mu_1}(\tau) h_{P_1,\mu_3}(\tau)h_{P_2,\mu_2}(\tau). \non \\
&&\qquad - \sum_{Q_1+Q_2=Q} \,q^{\frac{1}{2}Q^2-\frac{1}{2}
  (Q_1)^2_{2\cdot 1}-\frac{1}{2}(Q_2)^2_2}\, S(T_{(2\cdot 1 2)},t)\,2I_{12}\,
h_{2P_1,2\mu_1}(\tau)h_{P_2,\mu_2}(\tau).
\end{eqnarray}
The sum over $Q$ will give the correct mock Siegel theta functions (\ref{eq:Psi2})
and (\ref{eq:psi3}); the positivity condition of Subsection \ref{subsec:bpsmass} implies the 
convergence of the series. Note that for $P$ primitive, the semi-primitive
wall-crossing formula for $N=2$ is precisely such that modularity and 
integrality are compatible. This also suggests more generally, that the contribution to the
partition function from a rooted tree, based on a nested list of
magnetic charge, preserves $S$-duality. One will find products of vector-valued
modular forms corresponding to the different endpoints.

More evidence for the claim that $\bar \Omega(\Gamma)$ are the
correct invariants in the context of $S$-duality, can be found from the partition
functions of $\CN=4$ Yang-Mills on a surface \cite{Vafa:1994tf}, which are closely
related to D4-brane partition functions on a divisor of a
Calabi-Yau. These partition functions are
generating functions of the Euler number  $\chi(\CM)$ of the instanton moduli
space $\CM$,  which are related to the DT-invariants by 
$\Omega(\Gamma;t)=(-1)^{\dim \CM(\Gamma)}\chi(\CM(\Gamma))$ \cite{Diaconescu:2007bf}. 
 Yoshioka has calculated in Refs. \cite{Yoshioka:1994,
  Yoshioka:1995} the partition function for $U(2)$ Yang-Mills (rank 2 sheaves) on
$\mathbb{CP}^2$. The two partition functions for
sheaves of rank 2 with $c_1=0 \mod 2$ and $1\mod 2$ are given by  
\be
\label{eq:h20}
h_{2,0}(\tau)=-\frac{f_{2,0}(\tau)}{\eta(\tau)^6},
\qquad h_{2,1}(\tau)=\frac{f_{2,1}(\tau)}{\eta(\tau)^6},
\ee
where $f_{2,i}(\tau)$ are the generating functions of the class numbers $H(n)$:
\be
f_{2,0}(\tau)=\sum_{n=0}^\infty 3H(4n) q^{n},\qquad f_{2,1}(\tau)=\sum_{n=1}^\infty 3H(4n-1) q^{n-\frac{1}{4}}.
\ee
$(h_{2,0}(\tau),h_{2,1}(\tau))$ transforms as a vector-valued
modular form of weight $-\frac{3}{2}$.\footnote{The vector $(f_{2,0}(\tau),f_{2,1}(\tau)) $ is actually a
  mock modular form; a modular completion must be added for proper
  transformation properties under $SL(2,\mathbb{Z})$ \cite{Vafa:1994tf}.}  However, the coefficients of $h_{2,0}(\tau)$ are not
integers. To obtain integers, one needs to subtract the contribution of multiple
$U(1)$ instantons $\frac{1}{4}\frac{1}{\eta(2\tau)^3}$; the resulting
vector transforms only under $\Gamma_0(2)$. The
$-$-sign in (\ref{eq:h20}) is crucial and follows from the factor
$(-1)^{\dim \CM(\Gamma)}$. Similar results are known for $K3$
\cite{Vafa:1994tf, Minahan:1998vr}. 

Eq. (\ref{eq:contftrees}) suggests that the contribution of flow
trees to the index is most conveniently expressed in terms of $\bar
\Omega(\Gamma)$. This continues to be true for semi-primitive
wall-crossing with a larger multiplicity of $\Gamma_1$ and
non-primitive wall-crossing in general. Consider for example wall-crossing for
$(2\Gamma_1,2\Gamma_2)$. Eq. (\ref{eq:22}) expresses $\Delta
\bar \Omega(2\Gamma_1+2\Gamma_2;t)$ as a sum of terms indexed by nested lists which can be
attributed to different flow trees. It is not difficult to see that
this is a generic property of the jumps given by the
KS-formula. The non-trivial information provided by the KS-formula are
the prefactors of the contributions. Nested lists and flow trees are clearly useful tools
for enumerating invariants subject to wall-crossing.  

Of course, the integer invariants $\Omega(\Gamma)$ are useful too. 
For example, we have seen that the semi-primitive wall-crossing formula is a nice product formula in
terms of them. This has a geometric interpretation in terms
of halos ($N$ centers of $\Gamma_1$ placed on an equal distance around
a center with charge $\Gamma_2$), and correctly accounts for the
bose/fermi statistics \cite{Denef:2007vg}. 

One might wonder why $S$-duality and
integrality of the invariants are not compatible although they
are both well motivated from physics. A pragmatic reason is that
modularity seems to require that the jumps of the
indices can be written in terms of products of invariants, such that
the sum of the arguments of the invariants equals the total
charge. Such an identification is possible for  $\bar \Omega(\Gamma;t)$ but not for
$\Omega(\Gamma;t)$. 

Another physical motivation for the rational invariants are IIB
D-brane instantons. The IIA BPS-states can be
mapped to IIB instantons by a timelike T-duality, which
suggests that the instanton numbers are equal to the BPS-invariants
$\Omega(\Gamma;t)$. The invariants $\bar \Omega(\Gamma;t)$ appear for
instantons in their measure \cite{Pioline:2009ia}, the sum over
$m|\Gamma$ incorporates the contributions of multiple
instantons. This sum appears for D1-D(-1) instantons in fact after a Poisson resummation of a
manifestly $S$-duality invariant sum (analogous to Poincar\'e series)
\cite{RoblesLlana:2006is,RoblesLlana:2007ae}. The relation 
between $\Omega(\Gamma;t)$ and $\bar \Omega(\Gamma;t)$ is analogous to
Gromov-Witten invariants of $m$-fold covers of worldsheet instantons $\bar n_{Q,g}=\sum_m 
\frac{n_{Q/m,g}}{m^3}$, where $n_{Q,g}$ are also expected to be integers
\cite{Aspinwall:1991ce}. The rational invariants raise the question about the status of the MSW CFT for
non-primitive magnetic charges $P$. If this is a proper CFT, the
modular invariant partition function must have integer
coefficients. However, since the BPS-object is not protected by
conservation laws against decomposition into smaller objects, the
degrees of freedom might not combine to a proper conformal field theory.
  
\section{Summary and discussion}
\label{sec:conclusion}
The previous sections discussed the KS wall-crossing formula and flow trees,
and applied these to D4-D2-D0 black holes. Two new results which are
generally applicable to BPS wall-crossing using flow trees are: 
\begin{itemize}
\item[-] The sign of the flow parameter along every edge can be
  determined iteratively in terms of the initial moduli $t$, without
  explicit calculation of the flow throughout the tree. 

\item[-] It is demonstrated that $\Delta \bar \Omega(\Gamma;t)$ as
  derived from the KS-formula, can be decomposed into certain combinations of rational invariants $\bar
  \Omega(\Gamma,t)$ classified by nested lists, which also classify
  the flow trees. This suggests that the contribution to the
  index of a flow tree is conveniently expressed in terms of the rational
  invariants. 
\end{itemize}

The discussion on wall-crossing for D4-D2-D0 black holes is restricted to the large volume limit of  a single
K\"ahler cone. The following results are obtained:
\begin{itemize}
\item[-] For $N\leq 3$ is proven that the indefinite quadratic form
  $(Q-B)_+^2-\sum_{i=1}^N(Q_i-B)^2_i$ is positive definite for flow
  trees, since it is implied by the positivity of flow
  parameters in the tree. This result is expected to be true 
  for any $N$, which would imply that the BPS partition function in the mixed
  ensemble is convergent. 
\item[-] The contribution to the partition function of flow trees with
  3 endpoints is constructed, including the case where 2 endpoints have
  equal charge. The contribution of trees with non-primitive and primitive
  charges nicely combine to products of vector-valued
  modular forms, and mock Siegel theta functions.  
\item [-] The $S$-duality invariant partition function is a generating function of the
  rational invariants $\bar \Omega(\Gamma,t)$. It is conceivable that the contributions to
  the partition function of trees with prescribed magnetic charges preserve $S$-duality. 
\end{itemize}

Various aspects of wall-crossing for D4-D2-D0 BPS-states remain to be
better understood. A major aspect which was not addressed here, is the
modular completion of the mock Siegel theta function for $N=3$. This
prevented a confirmation of $S$-duality by the supergravity partition
function in this paper, although it is shown that important prerequisites are satisfied. The main
obstacles are 1) the signature of the indefinite quadratic form is 
$(2,3b_2-2)$, and 2) the complexity of the flow tree condition
$S(T_{(12)3},t)\neq 0$. The mathematical literature only reports on indefinite theta functions
and their modular completions for signature $(1,n-1)$. Another aspect
which deserves a better understanding is the physical interpretation
and derivation of the modular completion, it might
be related to perturbative contributions. Contributions to
the partition function of flow trees with $N>3$ are also left for
future research.  

This paper made various restrictions on the charges and the region of
moduli space; I hope to address in future research non-zero D6-brane
charge, to include finite volume effects and to cross walls between
K\"ahler cones. Another interesting direction is to
understand better the condition $S(T_{(12)3},t)\neq 0$ from a more mathematical
perspective, now it can be determined so easily in terms of $t$. An interesting
application in this context might be wall-crossing for sheaves on
surfaces as in \cite{Gottsche:1998}. 

\bigskip
\begin{center}{\bf Acknowledgements}\end{center}
It is my pleasure to thank Wu-yen Chuang, Atish Dabholkar, Emanuel
Diaconescu, Sheer El-Showk, Boris Pioline and Ashoke Sen for
discussion, and Frederik Denef for discussion and his comments on the manuscript.  This
work is supported in part by the ANR grant BLAN06-3-137168 and by the DOE grant
DE-FG02-96ER40949.

\providecommand{\href}[2]{#2}\begingroup\raggedright

\end{document}